\newif\ifreport

\reporttrue

\ifreport
\documentclass[12pt, draftclsnofoot, onecolumn]{IEEEtran}
\else
\documentclass[conference]{IEEEtran}
\fi

\usepackage{multirow}
\usepackage[Symbol]{upgreek}
\usepackage{booktabs}
\usepackage{threeparttable}
\usepackage{hyperref}
\usepackage{cite}
\usepackage{graphicx}
\usepackage{amsmath}
\usepackage{amssymb}
\usepackage{amsthm}
\usepackage{color}
\usepackage{algorithm}
\usepackage{algorithmic}
\theoremstyle{definition}
\newtheorem{theorem}{Theorem}
\newtheorem{lemma}{Lemma}

\newtheorem{corollary}{Corollary}

\newcommand{\tabincell}[2]{\begin{tabular}{@{}#1@{}}#2\end{tabular}}
\def\blu{\color{black}}
\def\blue{\color{black}}
\def\bluee{\color{black}}

\usepackage{setspace}

\IEEEoverridecommandlockouts

\begin{document}


\title{Analog Beam Tracking in Linear Antenna Arrays: Convergence, Optimality, and Performance}
\title{\huge Super Fast Beam Tracking in Phased Antenna Arrays}
\title{\huge \vspace{-4.5mm} Fast Beam Tracking in Phased Antenna Arrays: \\ Convergence, Optimality, and Performance}
\title{ Fast Analog Beam Tracking in Phased Antenna Arrays: Theory and Performance}

\author{Jiahui Li$^*$, Yin Sun$^\S$, Limin Xiao$^{\P\ddagger}$, Shidong Zhou$^*$, C. Emre Koksal$^\dagger$ \\
$^*$Dept. of EE, $^\P$Research Institute of Information Technology, Tsinghua University, Beijing, 100084, China\\
$^\S$Dept. of ECE, Auburn University, Auburn AL, 36849, U.S.A.\\
$^\dagger$Dept. of ECE, The Ohio State University, Columbus OH, 43210, U.S.A.
\thanks{$^\ddagger$Corresponding author.}
\thanks{J. Li, L. Xiao, and S. Zhou have been supported in part by National Basic Research Program of China (973 Program) grant 2012CB316002, National S\&T Major Project grant 2017ZX03001011-002, National Natural Science Foundation of China grant 61631013, National High Technology Research and Development Program of China (863 Program) grant 2014AA01A703, Science Fund for Creative Research Groups of NSFC grant 61321061, Tsinghua University Initiative Scientific Research grant 2016ZH02-3, International Science and Technology Cooperation Program grant 2014DFT10320, Tsinghua-Qualcomm Joint Research Program, and Huawei HIRP project. Y. Sun has been supported in part by ONR grant N00014-17-1-2417. C. Emre Koksal is supported in part by NSF grants CNS-1618566, CNS-1514260, and ONR grant N00014-16-1-2253. A part of this paper has been presented in the 2017 51st Asilomar Conference on Signals, Systems, and Computers \cite{li2017conf}.}
}

\maketitle

\begin{abstract}
The directionality of millimeter-wave (mmWave) communications introduces a significant challenge in serving fast-rotating/moving terminals, e.g., mobile AR/VR, high-speed vehicles, trains, UAVs.This challenge is exacerbated in mmWave systems using analog beamforming, because of the inherent non-convexity in the analog beam tracking problem. In this paper, we obtain the Cram\'er-Rao lower bound (CRLB) of beam tracking and optimize the analog beamforming vectors to get the minimum CRLB. Then, we develop a low complexity analog beam tracking algorithm that simultaneously optimizes the analog beamforming vector and the estimate of beam direction. Finally, by establishing a new basic theory, we provide the theoretical convergence analysis of the proposed analog beam tracking algorithm, which proves that the minimum CRLB of the MSE is achievable with high probability. Our simulations show that this algorithm can achieve faster tracking speed, higher tracking accuracy and higher data rate than several state-of-the-art algorithms. The key analytical tools used in our algorithm design are stochastic approximation and recursive estimation with a control parameter.


\end{abstract}

\bstctlcite{BSTcontrol}

\vspace{-5mm}
\section{Introduction}\label{sec_intro}
\vspace{-1mm}

In the upcoming 5th generation (5G) system, many high-throughput applications will be implemented to meet the increasing and diverse demands from users, such as mobile augmented/virtual reality (AR/VR)\cite{SunVR2018}, vehicle-to-vehicle/infrastructure (V2V/V2I) \cite{Wang2009V2V}, high-speed railway \cite{Wu2016Survey}, and unmanned aerial vehicle (UAV) \cite{Xiao2016Enabling}. Such explosively growing data traffic can be conveyed by using higher frequency bands, e.g., millimeter-wave (mmWave)\cite{Pi2011An, Boccardi2014Five, Heath2016overview}. Compared with traditional sub-6 GHz band, the propagation paths in the mmWave band are sparse, i.e., only the line-of-sight path and a few relatively strong reflected paths exist \cite{Rappaport2013Millimeter, Rappaport2015Wideband}. To ensure enough receive power, the antenna array with high array gain and narrow beam is needed \cite{Hang2010BeamAdapt}, and it is critical to track these paths fast and accurately in highly mobile conditions.

When using the large scale antenna array in the mmWave systems, it is not feasible to equip each antenna with a radio frequency (RF) chain and an A/D (or D/A) converter due to high energy consumption and huge hardware cost \cite{Heath2016overview}. Therefore, the hybrid (analog/digital) beamforming structure with phased antenna arrays were proposed as one of the most economical solutions \cite{Ohira2000Electronically, Sun2014Mimo, Han2015Large, Puglielli2016Design, Molisch2016Hybrid, Heath2016overview}. In this structure, the signals of all antennas are beamformed in the analog domain by using phase shifters, and one or several A/D (or D/A) converters are used for digital processing. This solution has been standardized by IEEE 802.11ad \cite{IEEE80211ad} and IEEE 802.15.3c \cite{IEEE802153c}, and is actively discussed by several 5G industrial organizations \cite{METIS2015, ITU2015}.

For each transmitter/receiver that is configured with a phased antenna array, {\bluee only a small amount of pilot symbols can be transmitted/received at a time}. This would lead to a high pilot training overhead for each terminal and thus less resource can be used for data transmission, which becomes even worse when the beam directions change rapidly. Hence, one fundamental challenge of mobile mmWave communication is how to accurately track the beam directions of enormous highly mobile terminals and achieve high data rate \& low pilot overhead, which has been recognized in the industry as one important research task for 5G mmWave and massive MIMO systems, e.g., \cite{Keysight2015massive, Samsung20155g, Amitava2016Enabling, Wen2016Bringing, Brown2016Promise}. This challenge is exacerbated in analog beamforming systems, where only one A/D (or D/A) converter is equipped. 

In this paper, we consider an analog beamforming system model and attempt to establish a basic theory of analog beam tracking, which has not been taken into consideration in the former studies. An efficient analog beam tracking algorithm is developed, which simultaneously achieves fast tracking speed, high tracking accuracy, high data rate, and low complexity. We prove that this algorithm converges to real beam directions with high probability, rather than other sub-optimal beam directions, at the maximum possible convergence rate for static scenarios. We also demonstrate that this translates into a highly accurate tracking of a large number of beams that are rotating at a high speed\footnote{These beams may come from either the same terminal or different terminals. For each terminal, the base station should keep track of several different beams to overcome the negative effect of channel blockage on some of the beam directions.}. Note that even with this simplified model, there still exist many difficulties to resolve the aforementioned challenge, which are summarized as follows:
\begin{itemize}
\item[1)] {\bluee Given a certain beam direction estimator, its estimation accuracy will be dependent on the pilot training beamforming vectors. Hence, the optimal beam tracking scheme should jointly optimize the analog beamforming vectors and the beam direction estimators.} 

\item[2)] Since only phase shifters are programmable for the analog phased antenna arrays, the optimization of analog beamforming vectors is a non-convex optimization problem.

\item[3)] Due to the multi-peak property of analog beam pattern, the problem of optimizing the beam direction estimators is non-convex, and there exist multiple locally optimal stable points.
\end{itemize} 
In addition, it is worth mentioning that the design concept and analytical framework proposed in this paper can potentially be generalized to more practical models, e.g., in one of our follow-up studies, a joint channel coefficient and beam direction tracking problem is solved \cite{Li2018mobilize}.

The main contributions and results of this paper are summarized as follows:

\begin{itemize}
\item We derive the CRLB of the beam tracking problem, which is a function of the analog beamforming vectors. Then, by optimizing among all analog beamforming vectors, we obtain the minimum CRLB. This minimum CRLB can help us find the optimal solution.

\item We develop an analog beam tracking algorithm for tracking dynamic beams, which jointly optimizes the analog beamforming vectors and beam direction estimators. The derived update equation of this algorithm is proportional to the imaginary part of the received signal, which makes the algorithm have very low complexity.




\item Theoretical convergence analysis of the proposed tracking algorithm is provided, showing that the minimum CRLB of the MSE is achievable under a certain condition. As the traditional stochastic approximation and recursive estimation theory (see \cite{nevel1973stochastic,kushner2003stochastic}) cannot be directly applied, we establish a new basic theory to prove the convergence and asymptotic optimality of our algorithm in \emph{static} beam tracking scenarios in three steps: First, we prove that it converges to a set of beam directions with probability one, including the real beam direction and some sub-optimal beam directions (Theorem \ref{th_convergence}). 
Second, we prove that under certain conditions, it converges to the real beam direction with high probability, rather than other sub-optimal beam directions (Theorem \ref{th_lock}). 
Finally, if the step-size parameters are chosen appropriately, then the mean squared error (MSE) of this algorithm converges to the minimum CRLB, and hence the highest convergence rate is achieved (Theorem \ref{th_normal}). To the extent of our knowledge, this paper presents the first theoretical analysis on the convergence and asymptotic optimality of the analog beam tracking problem.

\item Simulations in both \emph{static} and 
\emph{dynamic} beam tracking scenarios suggest that the proposed algorithm can achieve faster tracking speed, lower beam tracking error, and higher data rate than several state-of-the-art algorithms with the same pilot overhead, e.g., \cite{IEEE80211ad, KaramiLS2007, Gao2015multi, Alkhateeb2015Compressed, Rial2016Hybrid, Va2016tracking}. In particular, the numerical results in \emph{static} beam tracking scenarios reveal that this algorithm can converge quickly to the minimum CRLB, which verifies the theoretical performance limit we have derived and proved. In \emph{dynamic} beam tracking scenarios, at signal-to-noise-ratio (SNR) = 10 dB, the proposed algorithm is capable of tracking 18.33$^\circ$/s and achieve 95\% capacity by inserting only 5 pilots per sec (see Table \ref{tab_comp} in Section \ref{sec_dynamic_beam_tracking} for details), which corresponds to very fast tracking speed. In addition, this algorithm can support much lower SNR (e.g., -5 dB) and outperform the comparison algorithms a lot.
\end{itemize}

The rest of this paper is organized as follows. In Sections \ref{sec_relate_work}, we introduce the related work. In Sections \ref{sec_model}, the system model is described. In Sections \ref{sec_problem}, we formulate the beam tracking problem and obtain its performance bound. In Section \ref{sec_algorithm} and \ref{sec_analysis}, a recursive beam tracking algorithm is designed, which is proven to converge to the minimum CRLB in \emph{static} beam tracking scenarios. In Section \ref{sec_simulation}, numerical results show that this algorithm converges very fast to the minimum CRLB in the \emph{static} beam tracking scenarios and achieves a better tradeoff curve between MSE (or data rate) vs. angular velocity in \emph{dynamic} beam tracking scenarios.

\vspace{-4mm}
\section{Related Work}\label{sec_relate_work}
\vspace{-2.5mm}
\subsection{Beam Estimation and Tracking}\label{sec_related}
\vspace{-1.5mm}
There have been a large number of studies on beam direction estimation/tracking in mmWave systems with analog beamforming arrays. We first review the state-of-the-art algorithms:

\subsubsection{Beam estimation}This kind of methods sweep the channel with predefined spatial beams and estimate the beam directions of the channel based on these observations. According to different sweeping methods, we divide them into three categories: 1) Exhaustive sweeping \cite{Lee2014Exploiting, Payami2015Effective, zhu2016auxiliary}: Narrow spatial beams are used to probe the channel exhaustively. It can guarantee a thorough observation on the channel, but the pilot training overhead increases linearly with the number of antennas, which will easily go beyond the limitation of transmission resource. 2) Hierarchical multi-resolution sweeping \cite{Wang2009Beam, Hur2013Millimeter, Alkhateeb2014Channel, Alkhateeb2015Limited, Xiao2016Enabling}: The hierarchical multi-resolution codebooks are used to sweep the channel. 3) Random sweeping \cite{Gao2015multi, Alkhateeb2015Compressed, Rial2016Hybrid}: Several random analog beamforming vectors are used to observe the channel. Compared with the exhaustive sweeping methods, the latter two categories can reduce the pilot overhead a lot. However, the probing beams in these schemes might be far from the real beam directions, making the corresponding observation SNRs too low to contribute to the beam direction estimation. 
    
\subsubsection{Beam tracking} This kind of methods take the prior information of beam directions into account.
In \cite{zhang2016tracking, palacios2016tracking, Gao2016Fast, Va2016tracking}, the estimated beam directions are updated based on the latest estimates. In particular, the algorithms in \cite{zhang2016tracking, palacios2016tracking} use predefined probing directions to exhaustively sweep the channel in each iteration. Each time when the sweeping process is finished, the beam direction estimates will be updated recursively by using the latest observations and estimates. Due to the use of exhaustive sweeping in these algorithms, a large number of beam directions should be probed in each iteration, which introduces high pilot overhead. To avoid high pilot overhead, \cite{Gao2016Fast} proposed to probe several directions around the newly estimated beam directions, and \cite{Va2016tracking} used two pilots to probe two different directions within the mainlobe in each iteration. But the optimal probing directions are not given in these algorithms. In \cite{bae2017new}, the authors start to study the optimization of analog beamforming vectors during pilot training, which is obtained based on the latest estimate. However, its beam direction estimation is done without using the historical estimation information. 
In these studies, they did not simultaneously optimize the pilot training beamforming vectors and design the beam tracking scheme.
In addition, these works mainly used simulation results to verify the tracking performance without theoretical analysis on the convergence performance. However, as the beam tracking problem is non-convex, the convergence performance analysis is quite important.

\vspace{-3.5mm}
\subsection{Theoretical Analysis with Multiple Stable Points}\label{sec_related_theory}
\vspace{-0.5mm}

In \cite{koksal2012robust}, the authors analyzed the theoretical performance of their recursive SNR tracking algorithm. It proved that the algorithm can converge to the optimal solution asymptotically and approach the optimal Cram\'er-Rao lower bound (CRLB). In \cite{koksal2012robust}, the stochastic approximation and recursive estimation theory given in \cite{nevel1973stochastic} was used. However, this theory is mainly designed for the problem with a single stable point, which cannot handle the multi-stable-points beam tracking problem directly. In \cite{kushner2003stochastic}, the authors established a new theory that takes multiple stable points into consideration. A theorem was given to prove that a recursive process can converge to a unique point within the stable point set. But it still cannot handle the beam tracking problem, in which the algorithm should converge to the real beam direction, rather than other local optimal stable points. Because among all the stable points, only the real beam direction can provide an acceptable received SNR. To analyze whether a recursive process can converge to one particular stable point or not, \cite{borkar2008stochastic} has proposed a new theory, in which the lower bound of probability of convergence was derived. 

\vspace{-2mm}
\section{Model Description}\label{sec_model}
\vspace{-1mm}

\subsection{Notations} 
\vspace{-0.5mm}
Lower case letters such as $a$ and $\mathbf{a}$ is used to represent scalars and column vectors, respectively, where $|a|$ denotes the modulus of $a$ and $\|\mathbf{a}\|_2$ denotes the 2-norm of $\mathbf{a}$.  Upper case letters such as $\mathbf{A}$ will be utilized to denote matrices. For a vector $\mathbf{a}$ or a matrix $\mathbf{A}$, its transpose is denoted by $\mathbf{a}^\text{T}$ or $\mathbf{A}^\text{T}$, and its Hermitian transpose is denoted by $\mathbf{a}^\text{H}$ or $\mathbf{A}^\text{H}$. Let $\mathcal{CN}(u,\sigma^2)$ stand for the circular symmetric complex Gaussian distribution with mean $u$ and variance $\sigma^2$, and $\mathcal{N}(u,\sigma^2)$ stand for the real Gaussian distribution with mean $u$ and variance $\sigma^2$. The sets of (positive) integers and real numbers are written as $\mathbb{Z}(\mathbb{Z}^+)$ and $\mathbb{R}$, respectively. Expectation is denoted by $\mathbb{E}[\cdot]$ and the real (imaginary) part of a variable $x$ is denoted by $\operatorname{Re}\left\{ x \right\}$ $\left(\operatorname{Im}\left\{ x \right\}\right)$. The natural logarithm of $x$ is denoted by $\log(x)$. The phase of a complex number $z$ is obtained by $\angle z$.

\vspace{-2mm}
\subsection{System Model}
\vspace{-0.5mm}

Consider a receiver with a linear antenna array in Fig. \ref{fig_system}, where $M$ antennas are placed along a line, with a distance $d$ between adjacent antennas.\footnote{In practical systems, both transmitter and receiver sides should be considered. Due to the transmitter-receiver reciprocity, the beam tracking of both sides will have similar designs. Hence, we consider beam tracking algorithm design and performance analysis on the receiver side. In the following sections, in order to make the mathematical expressions and derivations more concise, the transmitter is simplified by using an omnidirectional antenna model.} The antennas are connected through programmable phase shifters to a single RF chain, and the phase shifters are controlled to steer the beam. In the mmWave band, the multi-path channel model is widely used \cite{Heath2016overview}. In addition, as block fading is often assumed in most literature, the channel vector is given by
\vspace{-1.5mm}
\begin{equation}\mathbf{h}_n = \sum_{\ell=1}^{N_\text{p}} \beta_{n,\ell} \mathbf{a}(x_{n,\ell}),\vspace{-1.5mm}\end{equation}
where $n$ is the time-slot number, $N_\text{p}$ is the number of distinct propagation paths, $\beta_{n,\ell}$ is the complex channel coefficient of the $\ell$-th path,	\vspace{-2mm}
\begin{equation}\label{eq_steer}
	\mathbf{a}(x_{n,\ell}) = \left[ 1~e^{j \frac{2\pi d}{\lambda} x_{n,\ell}}~\cdots~e^{j \frac{2\pi d}{\lambda}(M-1)x_{n,\ell}} \right]^\text{H},
	\vspace{-1mm}
\end{equation}
is the steering vector, $x_{n,\ell} = \sin(\theta_{n,\ell})$ is the normalized spatial frequency, $\theta_{n,\ell}\in[-\pi/2,\pi/2]$ is the angle-of-arrival (AoA), and $\lambda$ is the wavelength.

\begin{figure}
\centering
\vspace{-0mm}
\includegraphics[width=5.5cm]{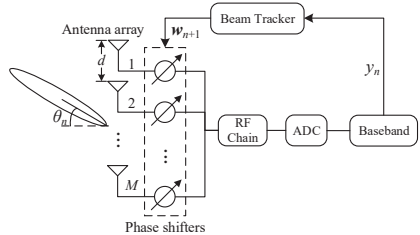}
\vspace{-2mm}
\caption{System model.}
\label{fig_system}
\vspace{-5mm}
\end{figure}

Thanks to the sparsity of mmWave channel and the use of large scale antenna array, the interference between different paths is weak in most cases \cite{Rappaport2013Millimeter, Rappaport2015Wideband}. Motivated by this, we assume that the single-path beam tracking is considered and multiple propagation paths can be tracked separately by using single-path beam tracking algorithms. This assumption was also adopted in recent studies, e.g., \cite{Va2016tracking}. Meanwhile, the constant-envelope constraint in phased antenna arrays (i.e., only phase shifters can be controlled in the analog beamforming vectors) makes beam tracking design much harder than the conventional fully digital arrays, which in turn renders the evaluation of fundamental limits and the development of achievable schemes extremely difficult. To that end, for analytical tractability in this paper, we will focus on the single-path case, in which one propagation path should be tracked.

For clarity, we omit the path number $\ell$. As depicted in Fig. \ref{fig_system}, the propagation path has an AoA $\theta_n$, and its channel vector is denoted by $\mathbf{h}_n = \beta_n \mathbf{a}(x_n)$. We assume that $\beta_n$ changes much slower than the beam direction $x_n$, hence it can be treated as a fixed and known parameter in the latter part, i.e., $\beta_n = \beta$.\footnote{ Regarding $\beta_n$, there exist two cases: (i) $\beta_n$ is assumed to known, and (ii) $\beta_n$ is unknown. In the first case, we can focus on the theoretical analysis of beam direction tracking, which will be considered in this paper. The second case is taken into consideration in our follow up paper \cite{Li2018mobilize}.}
Let $w_{mn} \in[-\pi,\pi]$ be the phase shift in radians provided by the $m$-th phase shifter. Then, the analog beamforming vector steered by the phase shifters is \vspace{-1mm}
\begin{equation}\label{eq_bf}
	\mathbf{w}_n = \frac{1}{\sqrt{M}}\left[ e^{jw_{1n}}~ e^{jw_{2n}}~ \cdots~e^{jw_{Mn}} \right]^\text{H}.
	\vspace{-1mm}
\end{equation}
Combining the output signals of the phase shifters yields \vspace{-1mm}
\begin{equation}\label{eq_receive}
	r_n =   \mathbf{w}_n^\text{H} \mathbf{h}_n p + \sigma {z}_n = p \beta\mathbf{w}_n^\text{H}\mathbf{a}(x_n) + \sigma {z}_n,
	\vspace{-1mm}
\end{equation}
where $p$ is the pilot signal, $\sigma^2$ is the noise power at each antenna, and the ${z}_n$'s are \emph{i.i.d.} circularly symmetric complex Gaussian random variables with zero mean and unity variance. Dividing $r_n$ in \eqref{eq_receive} by $p\beta$, the observation that contains the beam direction information is obtained as\vspace{-1mm}
\begin{equation}\label{eq_observation}
	y_n = \mathbf{w}_n^\text{H}\mathbf{a}(x_n) + \frac{{z}_n}{\sqrt{\rho}},\vspace{-1mm}
\end{equation}
where $\rho = {|p\beta|^2}/{\sigma^2}$ is the SNR at each antenna. Given $x_n$ and $\mathbf{w}_n$, the conditional probability density function of $y_n$ is\vspace{-1mm}
\begin{equation}\label{eq_pdf}
	p(y_n| x_n, \mathbf{w}_n) = \frac{\rho}{\pi} e^{- \rho\left| y_n -  \mathbf{w}_n^\text{H}\mathbf{a}(x_n) \right|^2}.\vspace{-1mm}
\end{equation}

A beam tracker determines the analog beamforming vector $\mathbf{w}_{n}$ and provides an estimate $\hat{x}_n$ of the spatial frequency $x_n$ after applying $\mathbf{w}_{n}$.\footnote{Interestingly, by tracking the spatial frequency $x_n$, we obtain a beam tracking algorithm with better robustness than tracking the AoA $\theta_n$; see Section \ref{further_discussion} for details.} From a control system perspective, $x_n$ is the system state, $\hat{x}_n$ is its estimate, the beamforming vector $\mathbf{w}_{n}$ is the control action, and $y_n$ is a noisy observation that is determined by a non-linear and non-convex function of the system state $x_n$ and control action $\mathbf{w}_{n}$. In the next section, we will formulate this beam tracking problem.






\vspace{-1mm}
\section{Beam Tracking and its Performance Bound}
\vspace{-1mm}
\label{sec_problem}


In Section \ref{sec_formulation}, we first formulate the beam tracking problem. Then, in Section \ref{sec_bound}, we derive a performance bound for the beam tracking problem.

\vspace{-2.5mm}
\subsection{Problem Formulation}\label{sec_formulation}
\vspace{-0.5mm}

{\blue Let $\psi\!=\!(\mathbf{w}_1, \mathbf{w}_2, \ldots, \hat{x}_1, \hat{x}_2, \ldots)$ represent a beam tracking policy, which contains a sequence of beamforming vectors $\{\mathbf{w}_n, n \in \mathbb{Z}^+\}$ and beam direction estimators $\{\hat{x}_n, n \in \mathbb{Z}^+\}$. In particular, we consider the set $\Psi$ of \emph{causal} beam tracking policies: The estimate $\hat{x}_n$ of time-slot $n$ and the control action $\mathbf{w}_{n+1}$ of time-slot $n+1$ are determined  by using the history of the control actions $(\mathbf{w}_1, \ldots, \mathbf{w}_n)$ and the observations $(y_1, \ldots, y_n)$. The policy $\psi$ is to be designed to minimize the beam tracking error in each time-slot. 
}
Given any time-slot $n$, the beam tracking problem can be formulated as \vspace{-1mm}
\begin{equation}\label{eq_problem}\vspace{-5mm}
	\blue \!\!\!\!\!\!\!\!\!\!\!\!\!\!\!\!\!\!\!\!\!\!\!\!\!\!\!\!\!\!\!\!\!\!\!\!\!\!\!\!\!\!\!\!\!\!\!\!\!\!\!\!\!\!\!\!\!\!\!\!\!\!\!\!\!\!\!\!\!\!\!\!\!\!\!\!\!\!\!\!\!\!\!\!\!\!\!\!\min_{\psi \in \Psi}~ \blue \mathbb{E}\left[ \left( \hat{x}_{n} - x_n \right)^2 \right]
	\end{equation}
	\begin{equation}\vspace{-1mm} \label{eq_constrant}
	\blue ~\text{s.t.}~ \blue \mathbb{E}\left[ \hat{x}_{n} \right] = x_n,  \mathbf{w}_n = \frac{1}{\sqrt{M}}\left[ e^{jw_{1n}}~ \cdots~ e^{jw_{Mn}} \right]^\text{H},
\end{equation}
where $\hat{x}_{n}$ in constraint \eqref{eq_constrant} is an \emph{unbiased} estimator of $x_n$.
Problem \eqref{eq_problem} is a constrained sequential control and estimation problem that is difficult, if not impossible, to solve optimally. First, the system is partially observed through the observation $y_n$. Second, both the control action $\mathbf{w}_{n}$ and the estimator $\hat{x}_n$ need to be optimized in Problem \eqref{eq_problem}: On the one hand, because only the phase shifts $(w_{1n}, \ldots, w_{Mn})$ in \eqref{eq_bf} are controllable, the optimization of $\mathbf{w}_{n}$ is a non-convex optimization problem. {\blue On the other hand, as the probability density function of $y_n$ is determined by the non-linear and non-convex function of $x_n$ as given in \eqref{eq_pdf}, the problem of optimizing the estimator $\hat{x}_n$ is also non-convex.} 

\vspace{-3mm}
\subsection{Lower Bound of Beam Tracking Error}\label{sec_bound}
\vspace{-1mm}
Next, we establish a lower bound of the MSE\footnote{\blue Note that as the unbiased estimator is considered, the MSE is equal to the variance of the proposed estimator. Therefore, we use MSE throughout the paper. In addition, the MSE is also a suitable performance metric for any estimators, including the biased and unbiased ones.} in \eqref{eq_problem} under the \emph{static} beam tracking scenarios, where $x_n=x$ for all time-slot $n$. Given the control actions $(\mathbf{w}_1, \ldots, \mathbf{w}_n)$, the MSE is lower bounded by the CRLB \cite{nevel1973stochastic}\vspace{-1mm}
\begin{equation}\label{eq_MMSE}
\mathbb{E}\left[ \left( \hat{x}_{n} - x \right)^2 \right] \ge \frac{1}{\sum_{i=1}^n I(x, \mathbf{w}_i)},\vspace{-1mm}
\end{equation}
where $I(x, \mathbf{w}_i) $ is the Fisher information \cite{Poor1994estimation} that can be computed by using \eqref{eq_pdf}:\vspace{-1mm}
\begin{equation}
\begin{aligned} 
	 I(x, \mathbf{w}_i) &= \mathbb{E}\left[ \left. - \frac{\partial^2 \log p \left( y_i|x, \mathbf{w}_i \right)}{\partial x^2} \right| x, \mathbf{w}_i \right] 
	= \frac{2\rho}{M}\left| \sum\limits_{m=1}^{M} \frac{2\pi d}{\lambda}(m-1) e^{j\left[w_{mi} - \frac{2\pi d}{\lambda}(m-1)x\right]} \right|^2.
\end{aligned}\vspace{-1mm}
\end{equation}
Note that the Fisher information $I(x, \mathbf{w}_i) $ is the function of $\mathbf{w}_i$. By optimizing the control actions $(\mathbf{w}_1, \ldots, \mathbf{w}_n)$ in the right-hand-side (RHS) of \eqref{eq_MMSE}, we obtain\vspace{-1mm}
\begin{equation}\label{eq_opt_MMSE}
\frac{1}{n} \sum_{i=1}^n I(x, \mathbf{w}_i) \le \frac{2M(M-1)^2\pi^2 d^2 \rho}{\lambda^2} \overset{\Delta}{=} I_{\max},\vspace{-1mm}
\end{equation}
where the maximum Fisher information $I_{\max}$ in (\ref{eq_opt_MMSE}) is achieved if, and only if, for $i = 1, \ldots, n$\vspace{-1mm}
\begin{equation}\label{eq_ctrl}
	\mathbf{w}_i = \frac{\mathbf{a}(x)}{\sqrt{M}} = \frac{1}{\sqrt{M}}\left[ 1~ e^{j \frac{2\pi d}{\lambda} x}~ \cdots~ e^{j \frac{2\pi d}{\lambda}(M-1)x} \right]^\text{H}.\vspace{-1mm}
\end{equation}
Hence, the MSE is lower bounded by the minimum CRLB\vspace{-1mm}
\begin{equation}\label{eq_CRLB}
	\mathbb{E}\left[ \left( \hat{x}_{n} - x \right)^2\right] \ge \frac{1}{nI_{\max}}.\vspace{-1mm}
\end{equation}

In what follows, we will investigate a new recursive analog beam tracking algorithm that can achieve this lower bound. 

\vspace{-2mm}
\section{Recursive Analog Beam Tracking Algorithm}\label{sec_algorithm}
\vspace{-0.5mm}
In this section, we will introduce our new recursive analog beam tracking algorithm, which tells how to recursively update the beamforming vectors and the estimates according to current observations and historical estimates. 

\vspace{-2.5mm}
\subsection{Frame Structure of Beam Tracking}\label{sec_alg_description}
\vspace{-0.5mm}

We first introduce the frame structure of the transmitted signals. The transmission is divided into two stages: 1) coarse beam sweeping and 2) recursive beam tracking. As depicted in Fig. \ref{fig_frame}, $M$ pilots will be received successively in \emph{Stage 1}, which is assumed to obtain an initial estimate $\hat{x}_0$. In \emph{Stage 2}, one pilot is allocated in each time-slot (e.g., at the beginning of each time-slot as in Fig. \ref{fig_frame}), and the estimate $\hat{x}_n$ as well as the control action $\mathbf{w}_n$ are updated iteratively.

\begin{figure}
\centering
\vspace{-2mm}
\includegraphics[width=5.8cm]{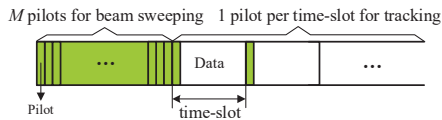}
\vspace{-2mm}
\caption{Frame structure.}
\vspace{-1mm}
\label{fig_frame}
\end{figure}

\begin{algorithm}[t]
\caption{Recursive Analog Beam Tracking}
\label{alg_1}
\begin{algorithmic}
\STATE {\begin{itemize}
\item[1)] \textbf{Coarse Beam Sweeping:} 
Receive $M$ pilots successively. The analog beamforming vector $\tilde{\mathbf{w}}_{m}$ for receiving the $m$-th training signal $\tilde{y}_m$ is\vspace{-1mm}
\begin{equation}\label{eq_codebook}
\tilde{\mathbf{w}}_{m} = \frac{\mathbf{a}\left(\frac{2m}{M} - \frac{M+1}{M}\right)}{\sqrt{M}}, m = 1, \ldots, M.\vspace{-1mm}
\end{equation}
Obtain the initial estimate  $\hat{x}_0$ of the beam direction by\vspace{-1.5mm} {\blue
\begin{equation}\label{eq_initial}
\begin{aligned}
	\hat{x}_0 =\underset{\hat{x} \in \mathcal{X}}{\arg\max}~\left|\mathbf{a}(\hat{x})^\text{H}\tilde{\mathbf{W}}\tilde{\mathbf{y}}\right|,
\end{aligned}\vspace{-1.5mm}
\end{equation}
where $\tilde{\mathbf{y}} = [\tilde{y}_1~ \cdots~ \tilde{y}_M]^\text{T}$, $\tilde{\mathbf{W}} = [\tilde{\mathbf{w}}_1~ \cdots~ \tilde{\mathbf{w}}_M]$}, and $\mathcal{X} = \left\{\frac{1 - M_0}{M_0}, \frac{3 - M_0}{M_0}, \ldots, \frac{M_0-1}{M_0}\right\}$.
\vspace{2mm}
\item[2)] \textbf{Recursive Beam Tracking:} In each time-slot $n = 1, 2, \ldots $, the analog beamforming vector $\mathbf{w}_{n}$ is\vspace{-1mm}
\begin{equation}\label{eq_est_ctrl}
	\begin{aligned}
	\mathbf{w}_{n} = \frac{\mathbf{a}\left(\hat{x}_{n-1}\right)}{\sqrt{M}}. 	
	\end{aligned}\vspace{-1mm}
\end{equation}
The estimate $\hat{x}_{n}$ of the beam direction is updated by\vspace{-1mm}
\begin{equation}\label{eq_est}
\begin{aligned}
\hat{x}_{n} = \left[ \hat{x}_{n-1} - a_n \operatorname{Im}\left\{ y_{n} \right\} \right]_{-1}^1,
\end{aligned}\vspace{-1mm}
\end{equation}
where $[x]_{b}^{c} = \max\left\{ \min\{ x, c \}, b \right\}$ and $a_n > 0$ is the step-size that will be specified later.
\end{itemize}
}
\end{algorithmic}
\end{algorithm}

\vspace{-2.5mm}
\subsection{Algorithm Design}\label{sec_ration}
\vspace{-0.5mm}

Based on the frame structure introduced in Section \ref{sec_alg_description}, we design a recursive analog beam tracking algorithm as described in Algorithm \ref{alg_1}. Then, we will clarify how Algorithm \ref{alg_1} is designed. Due to the non-convex property of the beam tracking problem in \eqref{eq_problem}, a good initial estimate $\hat{x}_0$ obtained in \emph{Stage 1} is quite important for the success of tracking the real direction $x_n$ in \emph{Stage 2}. As depicted in Fig. \ref{fig_imagpart}, we can conjecture that a good initial estimate $\hat{x}_0$ should be within the mainlobe set $\mathcal{B}\left(x_0\right)$, defined by\vspace{-1mm}
\begin{equation}\label{eq_mainlobe}
	\mathcal{B}\left(x_0\right) =\left( x_0 - \frac{\lambda}{Md}, x_0 +  \frac{\lambda}{Md}\right) \bigcap [-1, 1].\vspace{-1mm}
\end{equation}
To achieve this goal, the exhaustive sweeping is used to thoroughly probe the channel (i.e., by using the beamforming vectors in \eqref{eq_codebook}), and {\blue  motivated by the orthogonal matching pursuit method (e.g., \cite{Alkhateeb2015Compressed}), the initial estimate $\hat{x}_0$ is obtained by projecting the observation vector on a redundant dictionary $\mathcal{X}$ in \eqref{eq_initial}, where the size $M_0$ of $\mathcal{X}$ determines the estimation resolution and a larger size $M_0$ provides a more accurate initial estimate.\footnote{\blue The observation vector can be denoted by $\tilde{\mathbf{y}} = \tilde{\mathbf{W}}^\text{H}\mathbf{a}(x_0) + \tilde{\mathbf{z}}$, where $\tilde{\mathbf{z}}$ is the observation noise vector. Since $\tilde{\mathbf{W}}$ is a unitary matrix, i.e., $ \tilde{\mathbf{W}}\tilde{\mathbf{W}}^\text{H} = \mathbf{I}_M$, we can get $\left|\mathbf{a}(\hat{x})^\text{H}\tilde{\mathbf{W}}\tilde{\mathbf{y}}\right| = \left|\mathbf{a}(\hat{x})^\text{H}\mathbf{a}(x_0) + \mathbf{a}(\hat{x})^\text{H}\tilde{\mathbf{W}}\tilde{\mathbf{z}}\right|$, which is maximized by choosing an $\hat{x}$ that is close to $x_0$. When $\tilde{\mathbf{z}} = \mathbf{0}$, the best choice is $\hat{x} = x_0$.}} Our simulations suggest that, if the SNR $\rho \geq 0$ dB and $M_0 = 2M$, a good initial estimate $\hat{x}_0$ within the mainlobe $\mathcal{B}\left(x_0\right)$ can be obtained with a probability higher than $99.99\%$.\footnote{We can use more time-slots (pilot resources) to support lower SNR in \emph{Stage 1}. As \emph{Stage 1} is executed only once, this will not increase the total pilot overhead by much.}



In \emph{Stage} 2, the recursive beam tracker in \eqref{eq_est} is motivated by the following maximum likelihood (ML) estimator:\vspace{-1mm}
\begin{equation}\label{eq_ML_estimator}
\underset{\hat{x}_n}{\max}\left\{\underset{\mathbf{w}_n}{\max}~\sum_{i=1}^n\mathbb{E}\bigg[\log p\left( y_i|\hat{x}_n,\!\mathbf{w}_i\!\right)\bigg| \begin{matrix} \hat{x}_n,\!\mathbf{w}_1,\!\ldots,\!\mathbf{w}_i, \\ \!y_1,\!\ldots,\!y_{i-1}\end{matrix}\!\bigg]\right\}\!,
\vspace{-1mm}\end{equation}
where $\hat{x}_n \in [-1, 1]$ and $\mathbf{w}_n$ is subject to \eqref{eq_bf}. We propose a two-layer nested optimization algorithm to solve \eqref{eq_ML_estimator}:

In the inner layer, to achieve the maximum value, it is equivalent to maximize the Fisher information to find the best control action $\mathbf{w}_n$ as follows:\vspace{-1mm}
\begin{equation}\label{eq_maxFI}\begin{aligned}
\underset{\mathbf{w}_n}{\max}~~&I(\hat{x}_{n-1}, \mathbf{w}_n) \\
\text{s.t.}~~&{\blue \mathbf{w}_n = \frac{1}{\sqrt{M}}\left[ e^{jw_{1n}}~ \cdots~ e^{jw_{Mn}} \right]^\text{H}}.
\end{aligned}\vspace{-1mm}\end{equation}
According to \eqref{eq_opt_MMSE}, the solution of \eqref{eq_maxFI} is given by  $\mathbf{w}_n = {\mathbf{a}(\hat{x}_{n-1})}/{\sqrt{M}}$, i.e., \eqref{eq_est_ctrl}.

In the outer layer, rather than directly solving \eqref{eq_ML_estimator}, we propose to use the stochastic Newton's method, given by \cite{nevel1973stochastic}
\begin{equation}\begin{aligned}\label{eq_newton}\hat{x}_{n} = &~\left[ \hat{x}_{n-1} - s_{n} \cdot \frac{\frac{\partial \log p\left( y_n|\hat{x}_{n-1}, \mathbf{w}_n \right)}{\partial \hat{x}_{n-1}}}{\mathbb{E}\left[\left.\frac{\partial^2 \log p\left( y_n|\hat{x}_{n-1}, \mathbf{w}_n \right)}{\partial \hat{x}_{n-1}^2}\right| \hat{x}_{n-1}, \mathbf{w}_n \right]} \right]_{-1}^1 \\
 = &~\left[ \hat{x}_{n-1} + s_{n} \cdot \frac{\frac{\partial \log p\left( y_n|\hat{x}_{n-1}, \mathbf{w}_n \right)}{\partial \hat{x}_{n-1}}}{I(\hat{x}_{n-1}, \mathbf{w}_n)} \right]_{-1}^1,
\end{aligned}\end{equation}
where $s_n$ is the step-size, $[x]_{-1}^{1}\!=\!\max\left\{ \min\{ x, 1 \}, -1 \right\}$ constrains the estimates within the feasible region $[-1, 1]$, \vspace{-1mm}
\begin{equation}\label{eq_newton_1}
\begin{aligned}
\frac{\partial \log p\left( y_n|\hat{x}_{n-1}, \mathbf{w}_n \right)}{\partial \hat{x}_{n-1}}
= 2 \rho\operatorname{Re}\left\{ \left[y_n -  \mathbf{w}_n^\text{H}\mathbf{a}(\hat{x}_{n-1})\right]^\text{H} \cdot \mathbf{w}_n^\text{H}\frac{\partial \mathbf{a}(\hat{x}_{n-1})}{\partial \hat{x}_{n-1}}  \right\},
\end{aligned}\vspace{-1mm}
\end{equation}
and \vspace{-1mm}
\begin{equation} \label{eq_newton_2}\begin{aligned}
I(\hat{x}_{n-1}, \mathbf{w}_n) 
= \frac{2\rho}{M}\left| \sum\limits_{m=1}^{M} \frac{2\pi d}{\lambda}(m\!-\!1) e^{j\left[w_{mn}\!-\!\frac{2\pi d}{\lambda}(m\!-\!1)\hat{x}_{n-1}\right]} \right|^2.
\end{aligned}\vspace{-1mm}\end{equation}
By plugging \eqref{eq_est_ctrl}, \eqref{eq_newton_1} and \eqref{eq_newton_2} into \eqref{eq_newton}, we can obtain the following recursive beam tracker (see Appendix \ref{app_recursive_beam_tracker} for the detailed derivation)\vspace{-1.5mm}
\begin{equation}\label{eq_recursive_new}
\begin{aligned}
\hat{x}_{n} = \left[ \hat{x}_{n-1} -  \frac{\lambda s_n}{\sqrt{M}(M-1)\pi d}\cdot \operatorname{Im}\left\{ y_{n} \right\} \right]_{-1}^1.
\end{aligned}\vspace{-1.5mm}
\end{equation}
Let $a_n = {\lambda s_n}/[{\sqrt{M}(M-1)\pi d}]$ in \eqref{eq_recursive_new} be the new step-size, then we can obtain \eqref{eq_est}. Hence, even though the original algorithm in \eqref{eq_newton} is quite complicated, we are able to simplify it significantly, which greatly reduces the computational complexity of the algorithm.
\vspace{-2mm}
\section{Asymptotic Optimality Analysis}\label{sec_analysis}\vspace{-1mm}

In this section, we first present the key challenge faced by Algorithm \ref{alg_1}. Then, a series of three theorems will be developed to prove its asymptotic optimality in \emph{static} beam tracking, which helps resolve this challenge. Finally, we will investigate an alternative scheme that can be used to perform beam tracking.

\vspace{-3mm}
\subsection{Multiple Stable Points for Recursive Procedure} \label{sec_stable}\vspace{-1mm}
To obtain the points that the recursive procedure \eqref{eq_est_ctrl} and (\ref{eq_est}) might converge to, we will introduce its corresponding ordinary differential equation (ODE). Using \eqref{eq_observation} and \eqref{eq_est_ctrl}, the recursive beam tracker in (\ref{eq_est}) can also be expressed as\vspace{-1mm}
\begin{equation}\label{eq_gx}
\hat{x}_{n} = \left[ \hat{x}_{n-1} + a_n \left( f(\hat{x}_{n-1}, x_n) - \frac{\operatorname{Im}\left\{{z}_n\right\}}{\sqrt{\rho}} \right) \right]_{-1}^1,\vspace{-1mm}
\end{equation}
where function $f: \mathbb{R} \times  \mathbb{R}\mapsto \mathbb{R}$ is defined as\vspace{-1.5mm}
\begin{equation}\label{eq_fx}
	\begin{aligned}
	f(v, x_n) \overset{\Delta}{=} - \frac{1}{\sqrt{M}}\operatorname{Im}\left\{\mathbf{a}(v)^\text{H}\mathbf{a}(x_n)\right\}.
	\end{aligned}\vspace{-1.5mm}
\end{equation}
This recursive procedure can be seen as a noisy, discrete-time approximation of the following ODE \cite[Section 2.1]{borkar2008stochastic}\vspace{-1mm}
\begin{equation}\label{eq_ODE}
\frac{d \hat{x}(t)}{dt} \!=\!\left\{ \begin{array}{cl} \max\{f(-1, x_n),0\} &~\text{if}~\hat{x}(t) = -1 \\
f(\hat{x}(t),x_n) &~\text{if}~-1<\hat{x}(t) < 1 \\
\min\{f(1, x_n),0\} &~\text{if}~\hat{x}(t) = 1,\end{array}\right.\vspace{-1mm}
\end{equation}
with $t\geq0$ and $\hat{x}(0) = \hat{x}_0$. According to \cite{kushner2003stochastic, borkar2008stochastic}, the recursive procedure will converge to one of the stable points of the ODE \eqref{eq_ODE}. Here the stable point of the ODE \eqref{eq_ODE} is defined as a point $v_0$ that satisfies $f(v_0, x_n) = 0$ and $f_v'(v_0, x_n) < 0$, which means that any starting point from the neighbourhood of $v_0$ will make the ODE converge to $v_0$ itself.

\begin{figure}[!t]
\centering
\vspace{-2mm}
\includegraphics[width=5.5cm]{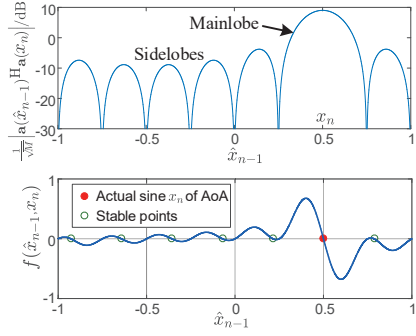}
\vspace{-2mm}
\caption{$\frac{1}{\sqrt{M}}\left|\mathbf{a}(\hat{x}_{n-1})^\text{H}\mathbf{a}(x_n)\right|$ and $f(\hat{x}_{n-1}, x_n)$ vs. $\hat{x}_{n-1}$ for $M = 8$, $x_n = 0.5$, $d = 0.5\lambda$. Note that the stable points are not at the sidelobe peaks.}
\vspace{-5mm}
\label{fig_imagpart}
\end{figure}

As depicted in Fig. \ref{fig_imagpart}, $f(v, x_n)$ is not monotonic in $v$ (i.e., Problem \eqref{eq_problem} is non-convex), and within each lobe (i.e., the mainlobe or the sidelobe) of the antenna array pattern, there exists one stable point. The \emph{local optimal stable points} for the recursive procedure are given by\vspace{-1mm}
\begin{equation}
\begin{aligned}\label{eq_stable_points}
\mathcal{S}(x_n)=&~\left\{ v \in (-1, 1]: f(v, x_n) =0, f_v'(v, x_n) <0\right\} \\
       =&~\left\{ v_k \in (-1, 1]: v_k = x_n + \frac{k\lambda}{(M-1)d}, k \in \mathbb{Z} \right\}.
\end{aligned}\vspace{-1mm}\end{equation}

{\blue 
Since the noise power is not zero, i.e., $\rho < \infty$, the recursive procedure may get out of the mainlobe and converge to one of the sub-optimal beam directions instead of the real beam direction, even though the initial value is within the mainlobe.} Note that except for $x_n$, the antenna array gain is quite low at other local optimal stable points in $\mathcal{S}(x_n)$, where the loss of antenna array gain is nearly 20~dB and will be higher if more antennas are configured. Hence, one key challenge is \emph{how to ensure that Algorithm \ref{alg_1} converges to the real direction $x_n$, instead of other local optimal stable points  in $\mathcal{S}(x_n)$.}

\vspace{-3mm}
\subsection{Step-size Design and Asymptotic Optimality Analysis}\label{sec_performance_analysis}\vspace{-0.5mm}
In \emph{static} beam tracking, we adopt the widely used diminishing step-sizes, given by \cite{nevel1973stochastic,kushner2003stochastic, borkar2008stochastic}
\begin{equation}\label{eq_stepsize}
a_n = \frac{\alpha}{n + N_0}, ~n = 1, 2, \ldots,
\end{equation}
where $\alpha > 0$ and $N_0 \ge 0$.

{\blue The traditional theory in stochastic approximation and recursive estimation does not provide an available method to prove the convergence to the optimal solution and the minimum CRLB under the condition of multiple stable points (see, e.g., \cite{nevel1973stochastic,kushner2003stochastic}). Therefore, we refer to the method that uses the ODE to obtain the lower bound of probability of convergence to any stable point \cite{borkar2008stochastic}. By using these tools, we propose a new basic theory to analyze Algorithm \ref{alg_1}. In particular, we now develop a series of three theorems to resolve the challenge mentioned in Section \ref{sec_stable}. }




\begin{theorem}[\textbf{Convergence to Stable Points}]\label{th_convergence}
\vspace{-1mm}
If $a_n$ is given by (\ref{eq_stepsize}) with any $\alpha > 0$ and $N_0 \ge 0$, then $\hat{x}_n$  converges to a unique point within $\mathcal{S}(x) \cup \{ -1\} \cup \{1 \}$ with probability one.
\vspace{-1mm}
\end{theorem}

%
%
%

\begin{proof}
See Appendix \ref{proof_converge}.
\end{proof}

%
%

Hence, for general step-size parameters $\alpha$ and $N_0$ in \eqref{eq_stepsize}, $\hat{x}_n$ converges to a stable point in $\mathcal{S}(x)$ or a boundary point. {\blue However, as mentioned in Section \ref{sec_stable}, among all these stable points, only the real beam direction $x$ can ensure the optimum antenna array gain, which is much higher than other local optimal stable points. Therefore, we further consider the case that $\hat{x}_n$ converges to the optimum solution $x$.}

\begin{theorem}[\textbf{Convergence to the Real Direction $x$}]\label{th_lock}
\vspace{-1mm}
If (i) the initial point satisfies $\hat{x}_0 \in \mathcal{B}\left(x\right)$, (ii) $a_n$ is given by (\ref{eq_stepsize}) with any $\alpha > 0$,
then there exist $N_0 \ge 0$ and $C_0>0$ such that\vspace{-1mm}
\begin{equation}\label{eq_lock}
P\left( \left. \hat{x}_n \rightarrow x \right| \hat{x}_0 \in \mathcal{B}\left(x\right) \right) \geq 1- 2e^{-C_0\cdot\frac{\rho}{\alpha^2}}.\vspace{-1mm}
\end{equation}


\vspace{-1mm}
\end{theorem}

\begin{proof}[Proof Sketch]
Motivated by Chapter 4 of \cite{borkar2008stochastic}, we will prove this theorem in three steps: in \emph{Step 1}, we will construct two continuous processes based on the discrete process $\left\{ \hat{x}_n \right\}$; in \emph{Step 2}, using these continuous processes, we establish a sufficient condition for the convergence of the discrete process $\left\{ \hat{x}_n \right\}$;  in \emph{Step 3}, we will derive the lower bound of probability for this condition, which is also a lower bound for $P\left( \left. \hat{x}_n\!\rightarrow\!x \right| \hat{x}_0\!\in\!\mathcal{B}\left(x\right) \right)$. See Appendix \ref{proof_lock} for the details.
\end{proof}

%
%

\ifreport
%
%
%
%

\fi

By Theorem \ref{th_lock}, if the initial point $\hat{x}_0$ is in the mainlobe $\mathcal{B}(x)$, the probability that $\hat{x}_n$ does not converge to $x$ decades \emph{exponentially} with respect to (w.r.t.) ${\rho}/{\alpha^2}$. Hence, we can increase the SNR $\rho$ and reduce the step-size parameter $\alpha$ to ensure $\hat{x}_n\!\rightarrow\!x$ with high probability. {\blu Under the condition of $\rho = 10~\text{dB}$ and $M = 8\text{-}128$, typical values of $N_0$ required by the sufficient condition in Theorem \ref{th_lock} are 10-50.} However, we can choose any $N_0\!\geq\!0$ to achieve a sufficiently high probability of $\hat{x}_n\!\rightarrow\!x$  in simulations.


\begin{theorem}[\textbf{Convergence to $x$ with the Minimum MSE}]
\vspace{-1mm}
\label{th_normal}
If (i) $a_n$ is given by (\ref{eq_stepsize}) with\vspace{-1mm}
\begin{equation}\label{eq_alpha}
\alpha = \frac{\lambda}{\sqrt{M}(M-1)\pi d} \overset{\Delta}{=} \alpha^*,\vspace{-1mm}
\end{equation}
and any $N_0 \ge 0$, and (ii) $\hat{x}_n \rightarrow x$,
then\vspace{-1mm}
\begin{equation}\label{eq_normal}
	\sqrt{n}\left(\hat{x}_n - x\right) \overset{d}{\rightarrow} \mathcal{N} \left(0, I_{\max}^{-1}\right),\vspace{-1mm}
\end{equation}
as $n \rightarrow \infty$, where $\overset{d}{\rightarrow}$ represents convergence in conditional distribution given $\hat{x}_n \rightarrow x$, and $I_{\max}$ is defined in \eqref{eq_opt_MMSE}. In addition,\vspace{-1mm}
\begin{equation}\label{eq_normal2}
	\lim_{n\rightarrow\infty}~n~\mathbb{E}\left[\left(\hat{x}_n - x\right)^2\big| \hat{x}_n \rightarrow x\right] = I_{\max}^{-1}.
\end{equation}
\vspace{-3mm}
\end{theorem}

\begin{proof}
See Appendix \ref{proof_normal}.
\end{proof}

%
%

Theorem \ref{th_normal} tells us that $\alpha$ should not be too small: If $\alpha=\alpha^*$ in \eqref{eq_alpha}, then the minimum CRLB on the RHS of \eqref{eq_CRLB} is achieved asymptotically with high probability, which ensures the highest convergence rate\footnote{The convergence rate is defined as the asymptotic properties of normalized errors \cite{Yin1994rate}, i.e., $\lim\limits_{n\rightarrow\infty}~n~\mathbb{E}\left[\left(\hat{x}_n - x\right)^2\right]$. Algorithm \ref{alg_1} is capable of approaching the minimum MSE, which corresponds to the highest convergence rate.}. In practice, we suggest to choose $\alpha=\alpha^*$ and $N_0=0$ in \eqref{eq_stepsize}.
Interestingly, Theorem \ref{th_normal} can be readily generalized to the track of any smooth function of $x$:
\begin{corollary}\label{co_1}
\vspace{-1mm}
If the conditions of Theorem \ref{th_normal} are satisfied, then for any first-order differentiable vector function $\mathbf{u}(v)$
\begin{align}
\!\!\lim_{n\rightarrow\infty} n\!~\mathbb{E}\!\left[\left\| \mathbf{u}(\hat{x}_n) - \mathbf{u}(x)\right\|^2_2\Big| \hat{x}_n \rightarrow x\right]\!=\!\left\|{ {\mathbf{u}}'(x) }\right\|^2_2 I_{\max}^{-1}.
\end{align}
\vspace{-4.5mm}
\end{corollary}
\begin{proof}
\ifreport
See Appendix \ref{proof_co_1}.
\else
See Appendix E in our technical report \cite{Li2017analog}. 
\fi
\end{proof}

For example, consider the channel response $\mathbf{h} = \beta \mathbf{a}(x)$ and its estimate $\hat{\mathbf{h}}_n =  \beta \mathbf{a}(\hat{x}_n)$. If $\alpha = \alpha^*$ and $N_0=0$,
Corollary \ref{co_1} tells us that, with a high probability, the minimum CRLB of $\mathbf{h}$ is achieved in the following limit:\vspace{-1mm}
\begin{equation}\label{eq_CRLB_CR}
\begin{aligned}
	\lim_{n\rightarrow\infty}~n\!~\mathbb{E}\left[\left.\left\| \hat{\mathbf{h}}_n - \mathbf{h}\right\|^2_2 \right| \hat{x}_n \rightarrow x \right]
	= I_{\max}^{-1} \sum_{m=1}^{M-1}\left|\frac{\partial \left(\beta e^{-j \frac{2\pi d}{\lambda}m x}\right)}{\partial x}\right|^2 = \frac{(2M-1)\sigma^2}{3(M-1)|p|^2}.
\end{aligned}\vspace{-1mm}
\end{equation}

\vspace{-3mm}
\subsection{Further Discussion: To Track the AoA $\theta$ or its Sine $x$?}\label{further_discussion}
\vspace{-0.5mm}
We can design the analog beam tracking algorithm by tracking either the AoA $\theta$ or its sine $x$. The  algorithm that tracks  the spatial frequency $x$ is provided in Algorithm \ref{alg_1}.  The  algorithm that directly tracks the AoA $\theta$ is described in Algorithm \ref{alg_2}.

\begin{algorithm}[t]
\caption{Angular Domain Recursive Analog Beam Tracking}
\label{alg_2}
\begin{algorithmic}
\STATE {\begin{itemize}
\item[1)] \textbf{Coarse Beam Sweeping:} 
Receive $M$ pilots successively. The analog beamforming vector $\tilde{\mathbf{w}}_{m}$ for receiving the $m$-th training signal $\tilde{y}_m$ is given by \eqref{eq_codebook}.
Obtain the initial estimate $\hat{\theta}_0$ of the beam direction by\vspace{-1.25mm}
\begin{equation}
\begin{aligned}
	\hat{\theta}_0\!=\!\arcsin\left\{\underset{\hat{x} \in \mathcal{X}}{\arg\max}~\left|\mathbf{a}(\hat{x})^\text{H}\tilde{\mathbf{W}}\tilde{\mathbf{y}} \right|\right\}.
\end{aligned}\vspace{-1.25mm}
\end{equation}
where $\tilde{\mathbf{y}} = [\tilde{y}_1~ \cdots~ \tilde{y}_M]^\text{T}$, $\tilde{\mathbf{W}} = [\tilde{\mathbf{w}}_1~ \cdots~ \tilde{\mathbf{w}}_M]$, and $\mathcal{X} = \left\{\frac{1 - M_0}{M_0}, \frac{3 - M_0}{M_0}, \ldots, \frac{M_0-1}{M_0}\right\}$.
\vspace{2mm}
\item[2)] \textbf{Recursive Beam Tracking:} In each time-slot $n = 1, 2, \ldots$, the analog beamforming vector $\mathbf{w}_{n}$ is\vspace{-1mm}
\begin{equation}\label{eq_est_theta_ctrl}
	\begin{aligned}
	\mathbf{w}_{n} = \frac{1}{\sqrt{M}} \mathbf{a}(\sin(\hat{\theta}_{n-1})).
	\end{aligned}\vspace{-1mm}
\end{equation}
The estimate $\hat{\theta}_{n}$ is updated by\vspace{-1mm}
\begin{equation}\label{eq_est_theta}
\begin{aligned}
\hat{\theta}_n = \left[ \hat{\theta}_{n-1} - \frac{a_n}{\cos(\hat{\theta}_{n-1})} \operatorname{Im}\left\{ y_{n} \right\} \right]_{-\frac{\pi}{2}}^{\frac{\pi}{2}},
\end{aligned}\vspace{-1mm}
\end{equation}
where $a_n > 0$ is the step-size.
\end{itemize}
}
\end{algorithmic}
\end{algorithm}

The convergence rate of Algorithm \ref{alg_2} can be characterized by Corollary \ref{co_1} with $u(x) = \arcsin x$. In particular, Algorithm \ref{alg_1} and Algorithm \ref{alg_2} share the same asymptotic convergence rate when $\hat{\theta}_{n}$ is very close to $\theta$. On the other hand,
if $\hat \theta_{n-1}$ is close to $-\frac{\pi}{2}$ or $\frac{\pi}{2}$, $\cos{(\hat \theta_{n-1})}$ in \eqref{eq_est_theta} is close to zero.
{\bluee Therefore, when $\theta_n$ is close to $-\frac{\pi}{2}$ or $\frac{\pi}{2}$, the update part in \eqref{eq_est_theta} has an infinite amplitude and Algorithm \ref{alg_2} will oscillate. However, this oscillation issue does not happen in Algorithm \ref{alg_1}.}

{\bluee Figure \ref{fig_x_vs_theta} depicts the tracking errors in angular degree in both algorithms, where the system parameters are configured as: $p\!=\!(1-j)/\sqrt{2}, \beta\!=\!(1+j)/\sqrt{2}, {\rho}\!=\!10~\text{dB}, M\!=\!8, d\!=\!0.5\lambda$, $\theta\!=\!88^\circ$, $x\!=\!\sin(\theta)\!\approx\!0.9994$, $a_n\!=\!{\alpha^*}/{10}$ or ${\alpha^*}$. It can be observed that both algorithms have similar tracking performance at the beginning. As the estimate gets closer to the real value, Algorithm \ref{alg_2} that tracks $\theta$ starts to oscillate or even escape the mainlobe, while Algorithm \ref{alg_1} is stable.}

\begin{figure}[!t]
\centering
\vspace{-2mm}
\includegraphics[width=14cm]{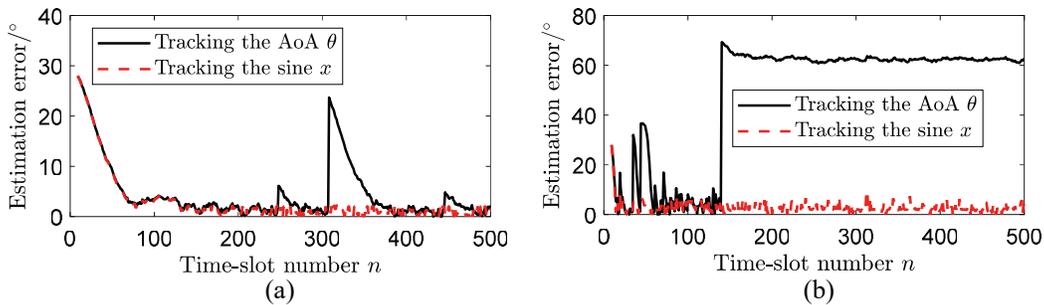}
\vspace{-5mm}
\caption{\bluee Estimation error comparison between the algorithms tracking the AoA $\theta$ and its sine $x$: (a) $a_n\!=\!{\alpha^*}/{10}$; (b) $a_n\!=\!{\alpha^*}$.}
\label{fig_x_vs_theta}
\vspace{-5mm}
\end{figure}

In addition, \eqref{eq_est_ctrl} and \eqref{eq_est} in Algorithm \ref{alg_1} are less complicated than \eqref{eq_est_theta_ctrl} and \eqref{eq_est_theta} in Algorithm \ref{alg_1} (although both algorithms are of low complexity). Because of these reasons, we choose to track the spatial frequency $x$ in this paper, instead of  tracking the AoA $\theta$ directly. If the AoA is needed, then one can use the arcsin function to obtain it, i.e., $\theta = \arcsin x$.

%

\vspace{-3mm}
\section{Numerical Analysis}\label{sec_simulation}
\vspace{-1mm}

We compare Algorithm \ref{alg_1} with four reference algorithms:
\begin{itemize}
\item[1)] \emph{Least squares \cite{KaramiLS2007}:} Sweep all the beamforming directions in the DFT codebook  (\ref{eq_codebook}) and use the least squares algorithm to estimate the channel response $\mathbf{h}_n$. Then obtain the analog beamforming vector $\mathbf{w}_{n}$ for data transmission by
\begin{align}
	w_{mn} = \angle \hat{h}_{nm}, m = 1, 2, \cdots, M,
\end{align}
where $\hat{h}_{nm}$ is the $m$-th element of the estimated channel response $\hat{\mathbf{h}}_n$.

\item[2)] \emph{Compressed sensing \cite{Gao2015multi, Alkhateeb2015Compressed, Rial2016Hybrid}:}
Randomly choose the phase shifts $w_{mn}$ from $\{ \pm 1, \pm j \}$ to receive pilot signals. Then use the sparse recovery algorithm to estimate the spatial frequency $x_n$, where a DFT dictionary with a size of 1024 is utilized.

\item[3)] \emph{IEEE 802.11ad \cite{IEEE80211ad}:} This algorithm contains two stages: beam sweep and beam tracking. In the first stage, sweep the beamforming directions in the DFT codebook (\ref{eq_codebook}) and choose the direction with the strongest received signal as the best beam direction. In the second stage, probe the best beam direction and its two adjacent beam directions, then choose the strongest direction as the new best beam direction. The second stage is performed periodically.

{\blue \item[4)] \emph{Kalman filter \cite{Va2016tracking}:}
We use the coarse beam sweeping method as given in Algorithm \ref{alg_1} to provide an initial estimate for the Kalman filter based algorithm. Then, in each iteration, two training beamforming vectors around the newly estimated AoA are used, where the observation angles' offsets are set as $\pm 3.5^\circ$. These offsets ensure that the observation angles are within the mainlobe. For fairness, we assume that $\beta$ is known in the Kalman filter based algorithm, and only the AoA $\theta_n$ need to be tracked.}

\end{itemize}
Two performance metrics are considered: (i) the MSE of the channel response $\mathbf{h}_n$, defined by\vspace{-1mm}
\begin{equation}
\begin{aligned}
\text{MSE}_{\textbf{h},n} \overset{\Delta}{=} \mathbb{E}\left[\left\| \hat{\mathbf{h}}_n - \mathbf{h}_n\right\|^2_2 \right],
\end{aligned}\vspace{-1mm}
\end{equation}
and (ii)
the achievable rate $R_n$, i.e.,\vspace{-1mm}
\begin{equation}
	R_n \overset{\Delta}{=} \log_2 \left( 1 + \rho\left|\textbf{{w}}^\text{H}_n\textbf{{a}}(x_n)\right|^2 \right).\vspace{-1mm}
\end{equation}
The system parameters are configured as: $p\!=\!(1-j)/\sqrt{2}, \beta\!=\!(1+j)/\sqrt{2}, {\rho}\!=\!10~\text{dB}, M\!=\!16, M_0\!=\!2M, d\!=\!0.5\lambda$. In the following subsections, we will investigate the static and dynamic beam tracking scenarios separately.

\vspace{-2mm}
\subsection{Static Beam Tracking}\vspace{-0.5mm}\label{static_channel_simulation}
In static beam tracking scenarios, we assume that one pilot is allocated in each time-slot. Hence, these algorithms have the same pilot overhead. The received pilot signals of all time-slots $1,\ldots, n$ are used for estimating $x_n$ and $\mathbf{h}_n$ in the compressed sensing and least square algorithms. The step-size $a_n$ is given by (\ref{eq_stepsize}) with $\alpha = \alpha^*$ and $N_0=0$. The simulation results are averaged over 10000 random system realizations, where the beam direction $x$ is randomly generated by a uniform distribution on $[-1, 1]$ in each realization.


Figure \ref{fig_static_mse} plots the convergence performance of $\text{MSE}_{\textbf{h},n}$ over time (i.e., pilot overhead). The MSE of Algorithm \ref{alg_1} converges quickly to the minimum CRLB given in \eqref{eq_CRLB_CR}, which agrees with Corollary \ref{co_1} and is much smaller than those of IEEE 802.11ad, least square, compressed sensing and Kalman filter algorithms.

\begin{figure}[!t]
\centering
\vspace{-3mm}
\includegraphics[width=6cm]{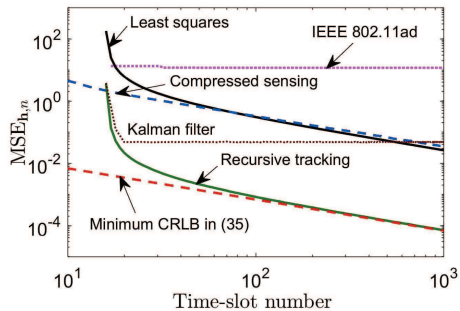}
\vspace{-3mm}
\caption{$\text{MSE}_{\textbf{h},n}$ vs. time-slot number $n$ in static beam tracking.}
\vspace{-5mm}
\label{fig_static_mse}
\end{figure}



\vspace{-2mm}
\subsection{Dynamic Beam Tracking}\vspace{-0.5mm}\label{sec_dynamic_beam_tracking}

In dynamic beam tracking scenarios, where beam direction changes over time. In the beginning, we assume that continuous pilot training is performed and an initial estimate is obtained for all the algorithms. After that, one pilot is allocated in each time-slot to ensure that these algorithms have the same amount of pilot overhead.

The last $M/2$ pilot signals are used in the compressed sensing algorithm and the last $M$ pilot signals are used in the least square algorithm. For the IEEE 802.11ad algorithm, the probing period of its beam tracking stage is 3 time-slots. For the Kalman filter algorithm, the probing period is 2 time-slots. These parameters are chosen to improve the performance of these algorithms.
To keep track of the changing beam direction, the step-size $a_n$ of Algorithm \ref{alg_1} is fixed as
\begin{equation}
a_n = \alpha^* =\frac{\lambda}{\sqrt{M}(M-1)\pi d},~\text{for all}~n \ge 1,
\end{equation}
which is determined by the configuration of the antenna array and is independent of the SNR $\rho$.

Figures \ref{fig_dynamic_tracking} and \ref{fig_dynamic_tracking_rate} depict the AoA tracking and achievable rate performance in dynamic scenarios, where the AoA $\theta_n$ varies according to $\theta_n\!=\!({\pi}/{3})\sin\left({2\pi n}/{1000}\right)\!+\!0.005\vartheta_n$ with $\vartheta_n\!\sim\!\mathcal{N}(0, 1)$. Algorithm \ref{alg_1} always tracks the actual AoA very well, and achieves the channel capacity $7.33$bits/s/Hz in all the time-slots. The performance of Algorithm \ref{alg_1} is much better than the first three reference algorithms, and the Kalman filter algorithm has similar performance as Algorithm \ref{alg_1}, which is due to that the current angular velocity is quite low. 

\begin{figure}[!t]
\centering
\vspace{-2mm}
\includegraphics[width=8cm]{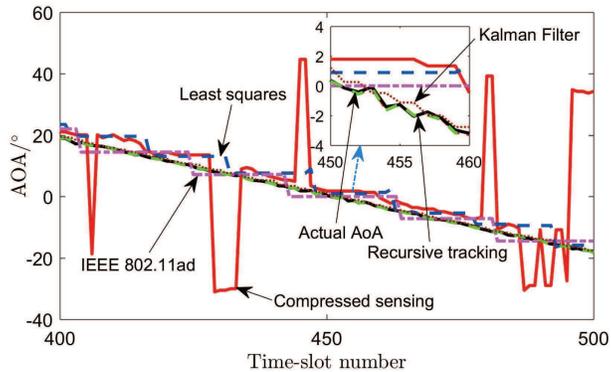}
\vspace{-4.5mm}
\caption{AoA tracking in dynamic beam tracking.}
\vspace{-2mm}
\label{fig_dynamic_tracking}
\end{figure}

\begin{figure}[!t]
\centering
\vspace{-0mm}
\includegraphics[width=8cm]{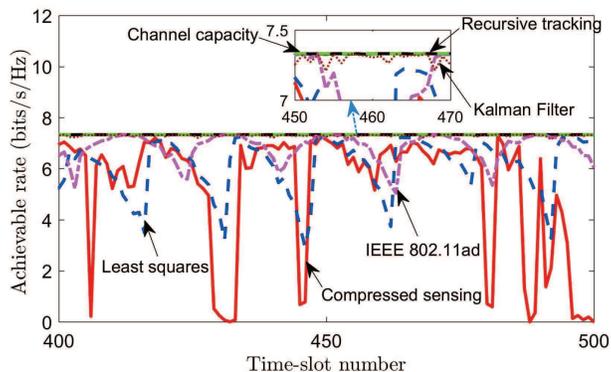}
\vspace{-4.5mm}
\caption{Achievable rate in dynamic beam tracking.}
\vspace{-4mm}
\label{fig_dynamic_tracking_rate}
\end{figure}

Next, we will consider different angular velocities. Figures \ref{fig_dynamic_MSEvsV} and \ref{fig_dynamic_RatevsV} illustrate the average AoA tracking and achievable rate performance under a fixed angular velocity model $\theta_n=\theta_{n-1}+\delta_{n-1}\!\cdot\!\omega$ where $1 \le n \le 10000$, $\theta_0=0$, $\delta_n\in\{-1,\!1\}$ denotes the rotation direction, and $\omega$ is a fixed angular velocity. The rotation direction $\delta_n$ is chosen such that $\theta_n$ varies within $[-{\pi}/{3},\!{\pi}/{3}]$. We can observe that Algorithm \ref{alg_1} can support higher angular velocities and data rates than the other algorithms when all 16 antennas are used. In addition, by using a subset of antennas, e.g., $M=4$ or 8, for beam tracking and all $16$ antennas for data transmission, the beam tracking regime of Algorithm \ref{alg_1} can be further enlarged.

According to Fig. \ref{fig_dynamic_RatevsV}, Algorithm \ref{alg_1} can achieve $95\%$ of the channel capacity when the angular velocity of the beam direction is 0.064 rad/time-slot, the SNR is $\rho=10~$dB, and $M=8$. If each time-slot (TTI) lasts for 0.2ms (e.g., in 5G systems \cite{Pedersen20165g, Zong20165g}), Algorithm \ref{alg_1} can support an angular velocity of $0.064 \times {1000}/{0.2} = 320~\text{rad/s} \approx 51~\text{circles/s}$. {\blue Consider a TDMA pilot pattern where 1000 beams are tracked by the antenna array periodically in a round-robin fashion such that 1 pilot is sent in each time-slot. Note that these beams may come from either the same terminal or different terminals, and the base station can keep track of several different beams for each terminal. Algorithm \ref{alg_1} can support \textbf{0.32 rad/s} (or \textbf{18.33$^\circ$/s}) per beam for tracking all these 1000 beams, which is \textbf{72~mph} if the transmitters/reflectors steering these beams are at a distance of 100 meters. And when it is needed to track extremely fast mobiles, we can insert pilot symbols with higher density (or frequency) for each mobile. For example, in the 5G NR standard, the subcarrier spacing of mmWave band is 60 kHz and each symbol spans 17.84 $\upmu$s \cite{Zaidi2016Waveform}. At this symbol rate, the proposed algorithm can track a narrow beam rotating at an angular velocity of $2.05\times 10^5$ degrees/sec. This corresponds to a rotation frequency of 411 Hz, which is way beyond the need of beam tracking. Hence, the tracking speed can be very fast. }

\begin{figure}[!t]
\centering
\vspace{-2mm}
\includegraphics[width=6.5cm]{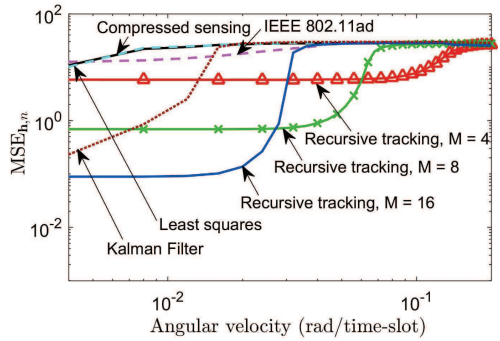}
\vspace{-3mm}
\caption{$\text{MSE}_{\textbf{h},n}$ vs. angular velocity in dynamic beam tracking, ${\rho} = 10~\text{dB}$.}
\vspace{-2mm}
\label{fig_dynamic_MSEvsV}
\end{figure}

\begin{figure}[!t]
\centering
\vspace{0mm}
\includegraphics[width=6.5cm]{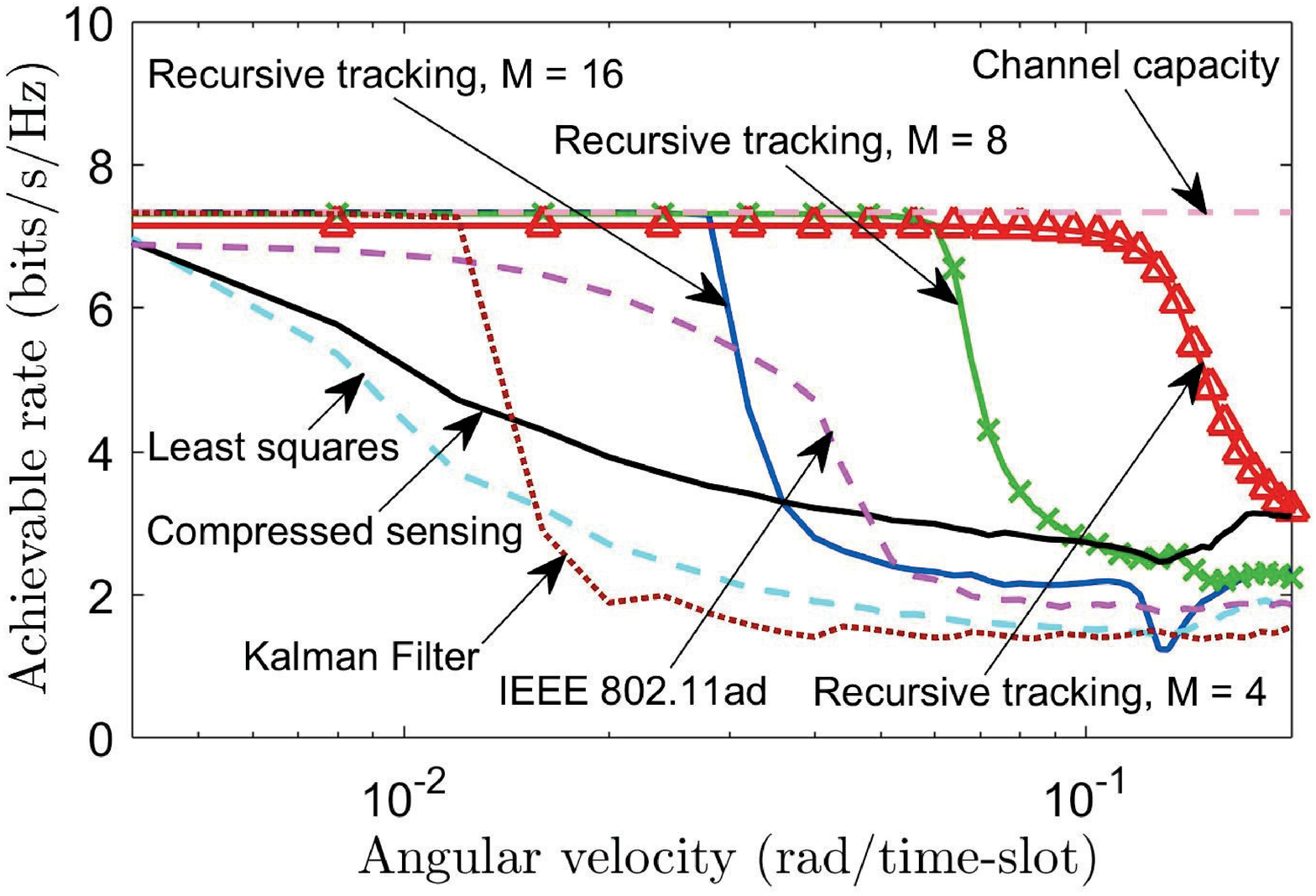}
\vspace{-3mm}
\caption{Achievable rate vs. angular velocity tradeoff in dynamic beam tracking,  ${\rho} = 10~\text{dB}$.}
\vspace{-4mm}
\label{fig_dynamic_RatevsV}
\end{figure}

In addition, we consider the condition that SNR is $\rho = 0~\text{dB}$ and other parameters are the same as Figs. \ref{fig_dynamic_MSEvsV} and \ref{fig_dynamic_RatevsV}. As depicted in Figs. \ref{fig_dynamic_MSEvsV_0dB} and \ref{fig_dynamic_RatevsV_0dB}, it can be seen that Algorithm \ref{alg_1} can provide higher performance gain than the condition that SNR is $\rho = 10~\text{dB}$, when all 16 antennas are used. Moreover, by using $M = 8$ antennas for tracking and all 16 antennas for data transmissions, the beam tracking regime of Algorithm \ref{alg_1} can still be enlarged. But when $M = 4$ antennas are used for tracking, the performance deterioration is quite significant due to the low antenna gain. Therefore, when SNR is low, more antennas are needed to ensure the good tracking performance.

\begin{figure}[!t]
\centering
\vspace{-0mm}
\includegraphics[width=6.5cm]{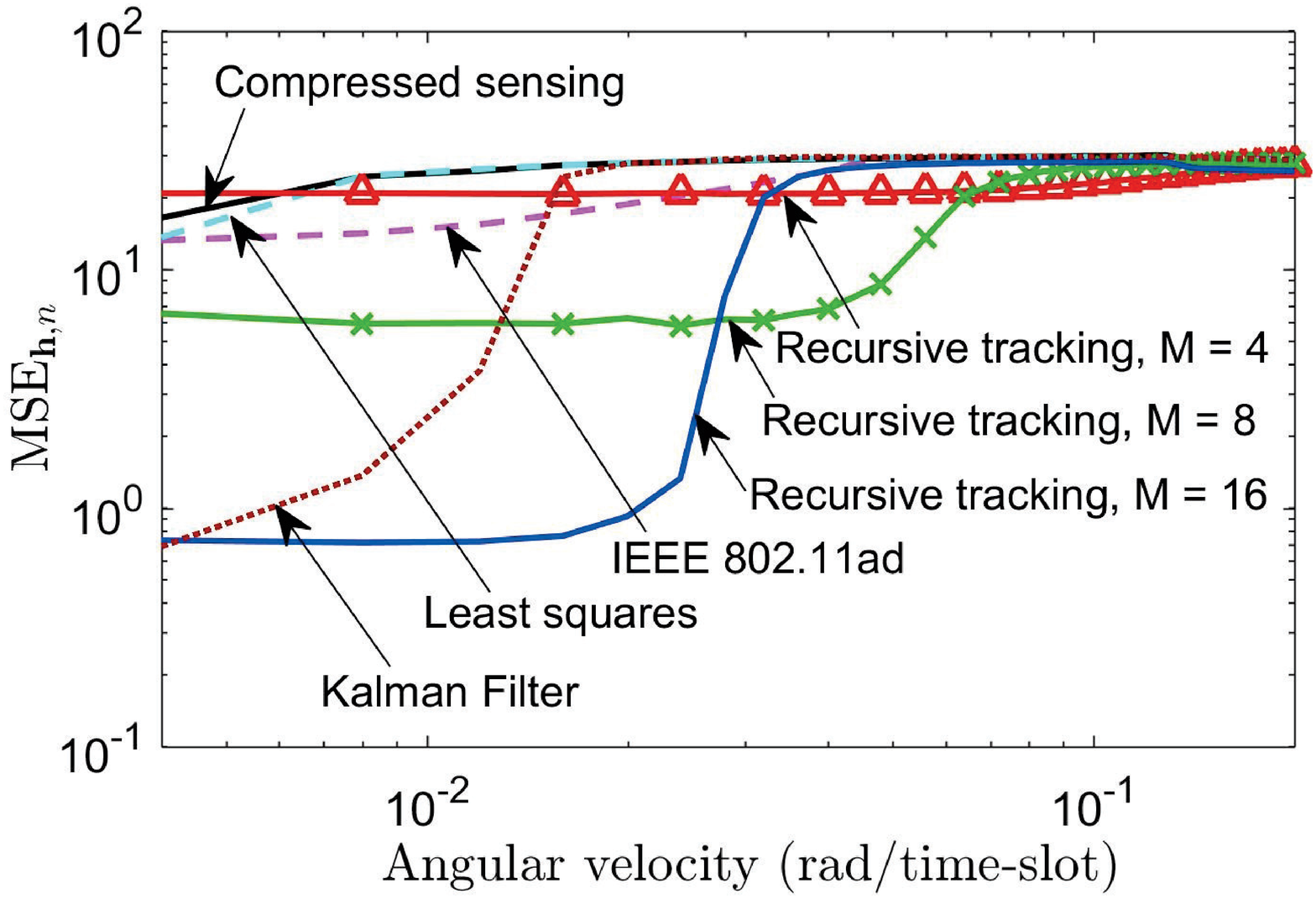}
\vspace{-3.5mm}
\caption{$\text{MSE}_{\textbf{h},n}$ vs. angular velocity in dynamic beam tracking, ${\rho} = 0~\text{dB}$.}
\vspace{-2mm}
\label{fig_dynamic_MSEvsV_0dB}
\end{figure}

\begin{figure}[!t]
\centering
\vspace{0mm}
\includegraphics[width=6.5cm]{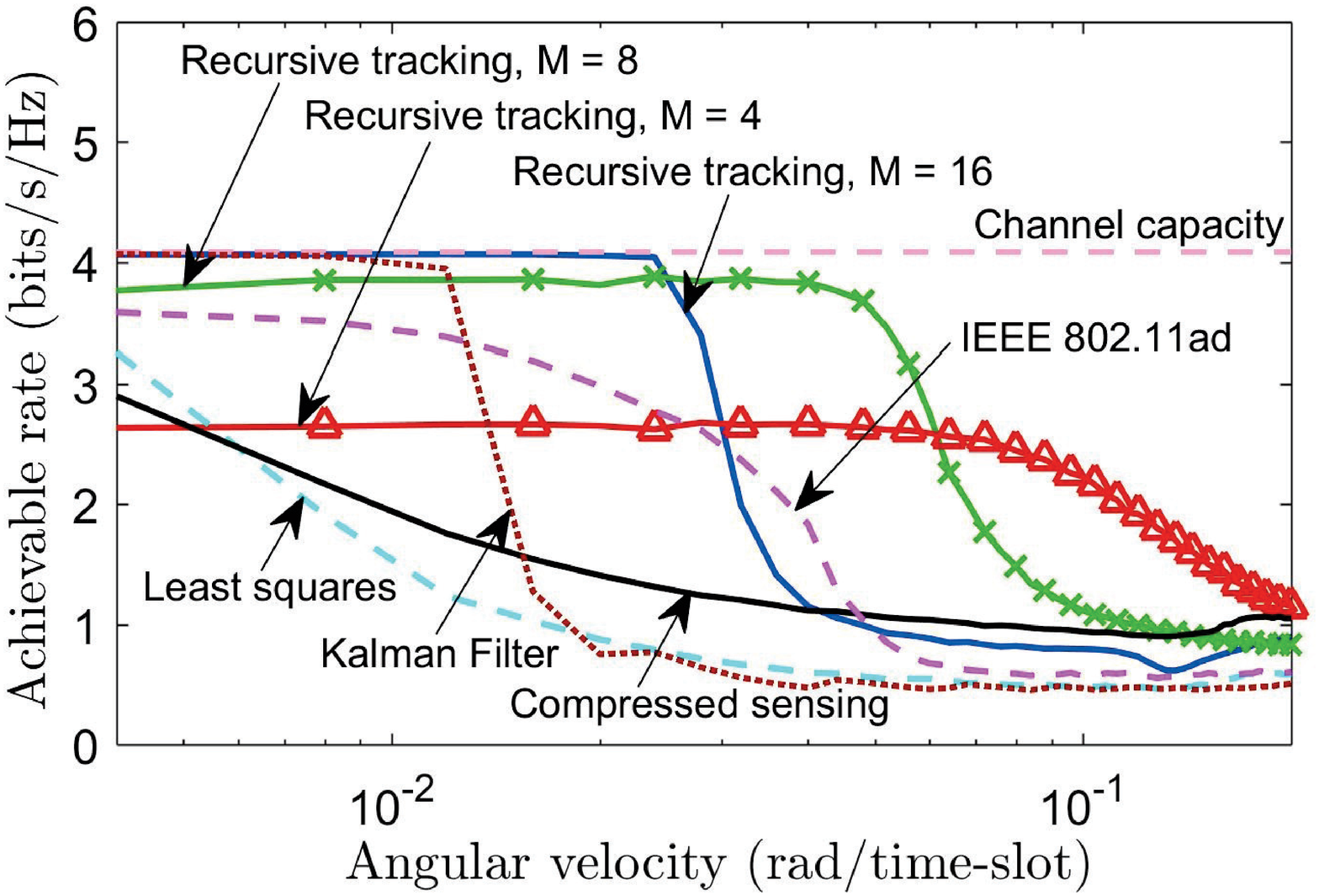}
\vspace{-3.5mm}
\caption{Achievable rate vs. angular velocity tradeoff in dynamic beam tracking, ${\rho} = 0~\text{dB}$.}
\vspace{-3mm}
\label{fig_dynamic_RatevsV_0dB}
\end{figure}

\begin{table}[!t]
\centering
\vspace{-0mm}
\caption{Maximum angular velocity (in {\blu \textnormal{degrees per second}}) for achieving $95\%$ of the channel capacity with different algorithms.}
\vspace{-2mm}
{\blue
\resizebox{0.425\textwidth}{!}{
\begin{tabular}{|c|c|c|c|c|c|c|}
\hline
{\tabincell{c}{\!\!\!SNR without \!\!\!\!\\ \!\!\! array gain  \!\!\!}} & {\tabincell{c}{\!\!\!Number of\!\!\!\! \\ antennas}} & {\tabincell{c}{\!\!\!Recursive \!\!\!\!\!\\ beam \\ tracking}} & {\tabincell{c}{\!\!Least\!\! \\ \!\!square\!\! \\ \cite{KaramiLS2007}}} & {\tabincell{c}{\!\!\!Compressed\!\!\! \\ sensing \\ \cite{Gao2015multi, Alkhateeb2015Compressed, Rial2016Hybrid}}} & {\tabincell{c}{IEEE \\ \!\!\!802.11ad\!\!\! \\ \cite{IEEE80211ad}}} & {\tabincell{c}{\!\!Kalman\!\! \\ filter \\ \cite{Va2016tracking}}} \\ \hline
\multirow{3}{*}{\!\!$10$~dB\!\!} & $M=8~$ & \textbf{18.33} & 4.13 & 2.29 & $-$ & $8.31$ \\ \cline{2-7}
& $M=32$ & {4.18} & 0.29 & 0.57 & 0.06 & $1.95$ \\ \cline{2-7}
& $M=128$ & {1.03} & 0.03 & 0.11 & 0.23 & $0.49$  \\ \hline
\multirow{3}{*}{\!\!$0$~dB\!\!} & $M=8~$ & \textbf{13.18} & $-$ & $-$ & $-$ & $6.02$ \\ \cline{2-7}
& $M=32$ & {3.84} & $-$ & $-$ & $-$ & $1.83$ \\ \cline{2-7}
& $M=128$ & {0.97} & $-$ & $-$ & $-$ & $0.46$ \\ \hline
\multirow{3}{*}{\!\!$-5$~dB\!\!} & $M=8~$ & $-$ & $-$ & $-$ & $-$ & $-$ \\ \cline{2-7}
& $M=32$ & {\textbf{2.98}} & $-$ & $-$ & $-$ & $0.29$ \\ \cline{2-7}
& $M=128$ & {0.92} & $-$ & $-$ & $-$ & $0.46$ \\ \hline\end{tabular} }
}
\begin{tablenotes}
        \footnotesize
        \item 1. The notation ``$-$'' denotes that the corresponding algorithm cannot achieve $95\%$ of the channel capacity even at zero angular velocity.
        \item 2. We assume that the SNR is the same for pilot training and data transmission, and 5 uniformly inserted pilot symbols per second are used for beam tracking.

\end{tablenotes}
\label{tab_comp}
\vspace{-6mm}
\end{table}

{\blue At last, the performance of these algorithms under different conditions is summarized in Table \ref{tab_comp}. We can observe that the maximum trackable angular velocity of Algorithm \ref{alg_1} to achieve 95\% capacity is much higher than those reference algorithms, and more importantly, when SNR is equal to or lower than 0~dB, Algorithm \ref{alg_1} still works well under most of the conditions, while the first three reference algorithms cannot meet the 95\% capacity requirement even if the moving speed is zero. Hence, the proposed algorithm can achieve a much faster beam tracking speed than the other algorithms, over a wide range of SNR values. }

\vspace{-2mm}
\section{Conclusions}\label{conclusion}
\vspace{-0.5mm}

We have developed an analog beam tracking algorithm, and established its convergence and asymptomatic optimality. Our theoretical and simulation results show that this algorithm can achieve faster tracking speed, lower beam tracking error, and higher data rate than several state-of-the-art algorithms. In our future work, we will consider hybrid beamforming systems with multiple RF chains, two-dimensional antenna arrays, and multi-path channel model with fast fading effects, based on the methodology developed in the current paper.

\ifreport
\appendices

{\bluee
\vspace{-2mm}
\section{Derivation of \eqref{eq_recursive_new}}\label{app_recursive_beam_tracker}
\vspace{-1mm}

By plugging \eqref{eq_est_ctrl} into \eqref{eq_newton_1}, we get
\vspace{-1mm}\begin{equation*}
\begin{aligned}
\frac{\partial \log p\left( y_n|\hat{x}_{n-1}, \mathbf{w}_n \right)}{\partial \hat{x}_{n-1}}
= &~2 \rho\operatorname{Re}\left\{ \left[y_n -  \mathbf{w}_n^\text{H}\mathbf{a}(\hat{x}_{n-1})\right]^\text{H} \cdot \mathbf{w}_n^\text{H}\frac{\partial \mathbf{a}(\hat{x}_{n-1})}{\partial \hat{x}_{n-1}}  \right\} \\
= &~2 \rho\operatorname{Re}\left\{ \left(y_n - \sqrt{M}\right)^\text{H} \cdot \frac{1}{\sqrt{M}} \left[\sum\limits_{m=1}^{M} -j\frac{2\pi d}{\lambda}(m\!-\!1)\right] \right\} \\
= &~2 \rho\operatorname{Re}\left\{ y_n^\text{H} \cdot \frac{1}{\sqrt{M}} \left[\sum\limits_{m=1}^{M} -j\frac{2\pi d}{\lambda}(m\!-\!1)\right] \right\} \\
= &~\frac{2\sqrt{M}(M-1)\pi d \rho}{\lambda}\cdot\operatorname{Re}\left\{ -j y_n^\text{H} \right\} = -\frac{2\sqrt{M}(M-1)\pi d \rho}{\lambda}\cdot\operatorname{Im}\left\{ y_n \right\}.
\end{aligned}\vspace{-1mm}
\end{equation*}

By plugging \eqref{eq_est_ctrl} into \eqref{eq_newton_2}, we can obtain
\vspace{-1mm}\begin{equation*}
\begin{aligned}
I(\hat{x}_{n-1}, \mathbf{w}_n) 
= &~\frac{2\rho}{M}\left| \sum\limits_{m=1}^{M} \frac{2\pi d}{\lambda}(m\!-\!1) e^{j\left[w_{mn}\!-\!\frac{2\pi d}{\lambda}(m\!-\!1)\hat{x}_{n-1}\right]} \right|^2 \\
= &~ \frac{2\rho}{M} \left| \sum\limits_{m=1}^{M} \frac{2\pi d}{\lambda}(m\!-\!1) \right|^2= \frac{2M(M-1)^2\pi^2 d^2 \rho}{\lambda^2}.
\end{aligned}\vspace{-1mm}\end{equation*}
Then, the update part in \eqref{eq_newton} is calculated as below
\begin{equation*}\begin{aligned}
\left.{\frac{\partial \log p\left( y_n|\hat{x}_{n-1}, \mathbf{w}_n \right)}{\partial \hat{x}_{n-1}}}\middle/{I(\hat{x}_{n-1}, \mathbf{w}_n)} = -\frac{\lambda}{\sqrt{M}(M-1)\pi d }\cdot\operatorname{Im}\left\{ y_n \right\} \right.,
\end{aligned}\end{equation*}
which leads to \eqref{eq_recursive_new}.

}

\vspace{-2mm}
\section{Proof of Theorem~\ref{th_convergence}}\label{proof_converge}
\vspace{-1mm}

Theorem 5.2.1 in \cite[Section 5.2.1]{kushner2003stochastic} provided the sufficient conditions under which $\hat{x}_n$ converges to a unique point within a set of stable points with probability one. We will prove that when the step-size $a_n$ is given by (\ref{eq_stepsize}) with any $\alpha > 0$ and $N_0 \ge 0$, our algorithm satisfies its sufficient conditions below:
\begin{itemize}
\item[1)] {\blue Step-size requirements: We can verify that $\lim\limits_{n\rightarrow \infty} a_n = 0$, $\sum\limits_{n=1}^\infty a_n = \infty$, and $\sum\limits_{n=1}^\infty a_n^2  < \infty$.}

\item[2)] It is needed to prove that $\sup\nolimits_n \mathbb{E} \left[ \left|-\operatorname{Im}\left\{ y_{n} \right\}\right|^2 \right] < \infty$. \vspace{0.5mm}
\item[] {\blue From (\ref{eq_gx}), for all $n \ge 1$, we have\vspace{-1mm}
\begin{equation}\vspace{-1mm}\label{eq_expectation_yn}\begin{aligned} \mathbb{E} \left[ \left|-\operatorname{Im}\left\{ y_{n} \right\}\right|^2 \right] 
= &~\mathbb{E} \left[ \left|f(\hat{x}_{n-1},x)\right|^2 + 2 f(\hat{x}_{n-1},x) \hat{z}_n + \hat{z}_n^2 \right] \\ \overset{(a)}{=} &~ \mathbb{E} \left[ \left|f(\hat{x}_{n-1},x)\right|^2 \right] + \frac{1}{2\rho} \\
\overset{(b)}{=} &~ \mathbb{E} \left[ \left|\frac{1}{\sqrt{M}}\sum_{m=1}^M e^{j\frac{2\pi d}{\lambda}(m-1)(\hat{x}_{n-1} - x)}\right|^2 \right] + \frac{1}{2\rho} 
\le M + \frac{1}{2\rho} < \infty, \end{aligned}\end{equation}
\item[] where $\hat{z}_n \overset{\Delta}{=} -{\operatorname{Im}\left\{{z}_n\right\}}/{\sqrt{\rho}} \sim \mathcal{N}\left( 0, {1}/({2\rho})\right)$, Step $(a)$ is due to that $\hat{z}_n$ is independent of $f(\hat{x}_{n-1}, x)$, and Step $(b)$ uses \eqref{eq_fx}. Hence, we get $\sup\nolimits_n \mathbb{E} \left[ \left|-\operatorname{Im}\left\{ y_{n} \right\}\right|^2 \right] < \infty$. } \vspace{0.5mm}

\item[3)] The function $f(v, x)$ should be continuous w.r.t. $v$.\vspace{0.5mm}
\item[] From \eqref{eq_fx}, $f(v, x)$ can be rewritten as follows:\vspace{-2mm}
\begin{equation*}\vspace{-2mm}
f(v, x) = - \frac{1}{\sqrt{M}}\sum_{m=1}^M \sin\left[\frac{2\pi d}{\lambda}(m-1)(v - x)\right].
\end{equation*}
Because $\sin\left[\frac{2\pi d}{\lambda}(m-1)(v - x)\right]$ is continuous w.r.t. $v$, and $f(v, x)$ is the summation of a finite amount of $\sin\left[\frac{2\pi d}{\lambda}(m-1)(v - x)\right], m = 1,\ldots, M$. Therefore, we can conclude that $f(v, x)$ is continuous w.r.t. $v$.
\vspace{0.5mm}
\item[4)] {\blue Let $\{ \mathcal{G}_n: n \ge 0 \}$ be an increasing sequence of $\sigma$-fields of $\left\{ \hat{x}_0, \hat{z}_{1}, \hat{z}_{2}, \ldots \right\}$, i.e., $\mathcal{G}_{n-1} \subset \mathcal{G}_n$, where $\mathcal{G}_0 \overset{\Delta}{=} \sigma(\hat{x}_0)$ and $\mathcal{G}_n \overset{\Delta}{=} \sigma(\hat{x}_0, \hat{z}_{1}, \ldots, \hat{z}_{n}) $ for $n \ge 1$. Define $\gamma_n \overset{\Delta}{=} \mathbb{E} \left[ \left.-\operatorname{Im}\left\{ y_{n} \right\}\right| \mathcal{G}_{n-1} \right] - f(\hat{x}_{n-1}, x)$. It is needed to prove that $\sum_{n=1}^\infty \left| a_n \gamma_n \right| < \infty$ with probability one.\vspace{0.5mm}
\item[] Since the $\hat{z}_n$'s are \emph{i.i.d.} circularly symmetric complex Gaussian random variables with zero mean, $\hat{z}_n$ is independent of $\mathcal{G}_{n-1}$, and $\hat{x}_{n-1} \in \mathcal{G}_{n-1}$, we have\vspace{-1mm}
\begin{align}\label{eq_fil2}
	\mathbb{E} \left[ \left. -\operatorname{Im}\left\{ y_{n} \right\} \right| \mathcal{G}_{n-1} \right] =&~ \mathbb{E} \left[ \left. f(\hat{x}_{n-1}, x) + \hat{z}_n \right| \mathcal{G}_{n-1} \right] \\
	= &~\mathbb{E} \left[ \left. f(\hat{x}_{n-1}, x)\right| \mathcal{G}_{n-1} \right] + \mathbb{E} \left[ \left. \hat{z}_n \right| \mathcal{G}_{n-1} \right]	= f(\hat{x}_{n-1}, x) \nonumber,
\end{align}
for $n \ge 1$. Hence, we can get $\gamma_n = 0$ for all $n \ge 1$ and $\sum_{n=1}^\infty \left| a_n \gamma_n \right| = 0 < \infty$ with probability one.}
\vspace{0.5mm}
\item[5)] The set of stable points for \eqref{eq_ODE} should be obtained.\vspace{0.5mm}
\item[] According to \eqref{eq_stable_points}, $\mathcal{S}(x)$ contains the local optimal stable points of the ODE \eqref{eq_ODE}. Also, the boundary point 1 (or $-1$) is a stable point when $f(1, x) \ge 0$ (or $f(-1, x) \le 0$). Hence, the set of stable points is $\mathcal{S}(x) \cup \{ -1\} \cup \{1 \}$.
\end{itemize}

By Theorem 5.2.1 in \cite{kushner2003stochastic}, $\hat{x}_n$ converges to a unique point within $\mathcal{S}(x) \cup \{ -1\} \cup \{1 \}$ with probability one. 

\vspace{-3mm}
\section{Proof of Theorem~\ref{th_lock}}\label{proof_lock}
\vspace{-1mm}

Theorem \ref{th_lock} is proven in three steps:

{\blue \noindent\emph{\textbf{Step 1:} Two continuous processes are constructed based on the discrete process $\left\{ \hat{x}_n \right\}$.}  

The first continuous process $\bar{x}(t), t \ge 0$ is the linear interpolation of the sequence $\{ \hat{x}_n \}$, where $\bar{x}(t_n) = \hat{x}_n, n \ge 0$ and $\bar{x}(t)$ is given by\vspace{-1.75mm}
\begin{equation}\vspace{-2.75mm}\label{eq_continuous}
\begin{aligned}
	\bar{x}(t)\!=\!\bar{x}(t_n)\!+\!\frac{(t\!-\!t_n)\left[\bar{x}(t_{n+1})\!-\!\bar{x}(t_n)\right]}{a_{n+1}}, t\!\in\![t_n, t_{n+1}],
\end{aligned}
\end{equation}
\vspace{-0.5mm}where $t_n$ is the discrete time parameter defined by $t_0 \overset{\Delta}{=} 0$, $t_n \overset{\Delta}{=} \sum_{i=1}^n a_{i}$, $n \ge 1$.

The second continuous process $\tilde{x}^n(t)$ is a solution of the ODE (\ref{eq_ODE}) for $t \in [t_n, \infty)$, where $\tilde{x}^n(t_n) = \bar{x}(t_n) = \hat{x}_n, n \ge 0$. Since we only care about the condition that $\hat{x}_n \in \mathcal{B}(x)$, the projection operation in the ODE (\ref{eq_ODE}) will not take effect and we can remove it\footnote{\blue There exist two cases: (i) if $\pm 1 \notin \mathcal{B}(x)$, then the solution of the ODE (\ref{eq_ODE}) is within $(-1, 1)$, (ii) if 1 (or $-1$) is in $\mathcal{B}(x)$, then $f(1, x) \le 0$ (or $f(-1, x) \ge 0$), hence the solution will not cross the boundary $\pm 1$.}. Then, we have $\frac{d\tilde{x}^n(t)}{dt} = f(\tilde{x}^n(t),x) $ and\vspace{-1mm}
\begin{equation}\vspace{-1mm}\label{eq_ODE_new}
\begin{aligned}
	\tilde{x}^n(t) & = \bar{x}(t_n) + \int_{t_n}^t f(\tilde{x}^n(v),x) dv, t \ge t_n.
\end{aligned}
\end{equation}

\noindent\emph{\textbf{Step 2:} Using the continuous processes $\bar{x}(t)$ and $\tilde{x}^n(t)$, a sufficient condition for the convergence of the discrete process $\left\{ \hat{x}_n \right\}$ is established.}

Let $\mathcal{I}$ be an invariant set that contains the real direction $x$ and is in the mainlobe, i.e., $x \in \mathcal{I} \subset \mathcal{B}(x)$.} Pick $\delta$ such that\footnote{The boundary of the set $\mathcal{B}(x)$ is denoted by $\partial \mathcal{B}(x)$.}\vspace{-1.5mm}
\begin{equation}\vspace{-0.5mm}\label{eq_delta_choice}
\inf_{v \in \partial \mathcal{B}(x)} \left| v - \hat{x}_0 \right| > \delta > 0.
\end{equation}
Then, the invariant set $\mathcal{I}$ can be constructed as follows:\vspace{-1mm}
\begin{equation}\vspace{-1mm}
	\mathcal{I} = \Big( x - |x - \hat{x}_0| - \delta,~x + |x - \hat{x}_0| + \delta \Big) \subset \mathcal{B}(x).
\end{equation}
An example of the invariant set $\mathcal{I}$ is illustrated in Fig. \ref{fig_invariant_set}.

\begin{figure}[!t]
\vspace{-7mm}
\centering
\includegraphics[width=4.5cm]{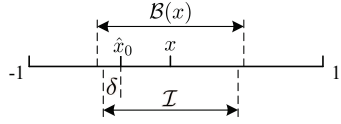}
\vspace{-3mm}
\caption{An illustration of the invariant set $\mathcal{I}$.}
\label{fig_invariant_set}
\vspace{-5mm}
\end{figure}

Next, we will establish a sufficient condition that ensures $\hat{x}_n\!\in\!\mathcal{I}~\text{for}~n\!\ge\!0$, and hence from Theorem 2.1 in \cite{borkar2002lock}, we can obtain that $\{\hat{x}_n\}$ converges to $x$.
Before giving the sufficient condition, let us provide some useful definitions first:
\begin{itemize}
\item Pick $T > 0$ such that the solution $x(t), t\ge 0$ of the ODE (\ref{eq_ODE}) with $x(0)\!=\!\hat{x}_0$ satisfies $\inf_{v \in \partial \mathcal{I}}\left| v\!-\!x(t) \right| > 2\delta$ for $t \ge T$. Since the solution $x(t)$ of the ODE (\ref{eq_ODE}) will approach the real direction $x$ monotonically within the mainlobe $\mathcal{B}(x)$ as time $t$ increases, we have a sufficient condtion of $\inf_{v \in \partial \mathcal{I}}\left| v\!-\!x(t) \right| > 2\delta, t \ge T$ as $\left| \hat{x}_0\!-\!x(T) \right| \ge \delta$, hence one possible $T$ is given by\vspace{-1mm}
\begin{equation}\vspace{-1mm}\label{eq_T}
T=\frac{\delta}{\min \left\{|f(\hat{x}_0, x)|, |f(|\hat{x}_0\!-\!x|\!-\!\delta\!+\!x, x)|\right\}},
\end{equation}
where the denominator is to obtain the minimum absolute gradient for $t \in [0, T]$.

\item Let $T_0 \overset{\Delta}{=} 0$ and $T_{m+1} \overset{\Delta}{=} \min \left\{ t_i: t_i \ge T_n + T, i \ge 0 \right\}$ for $m \ge 0$. Then $T_{m+1} - T_m \in [T, T+a_1]$ and $T_m = t_{\tilde{n}(m)}$ for some $\tilde{n}(m) \rightarrow \infty$, where $\tilde{n}(0) = 0$. Let $\tilde{x}^{\tilde{n}(m)}(t)$ denote the solution of ODE (\ref{eq_ODE}) for $t \in I_m \overset{\Delta}{=} \left[ T_m, T_{m+1} \right]$ with $\tilde{x}^{\tilde{n}(m)}(T_m) = \bar{x}(T_m)$, $m \ge 0$.
\end{itemize}
Then, we can obtain the following lemma:
\begin{lemma}\label{le_sufficient}
If $\underset{t\in I_m}{\sup} \left| \bar{x}(t) - \tilde{x}^{\tilde{n}(m)}(t)\right| \le \delta$ for all $m \ge 0$, then $\hat{x}_n \in \mathcal{I}~\text{for all}~n \ge 0$.
\end{lemma}
\begin{proof}
\ifreport
See Appendix \ref{sec_proof_le_sufficient}
\else
See Appendix F in our technical report \cite{Li2017analog}. 
\fi
\end{proof}

{\blue
\noindent\emph{\textbf{Step 3:} Based on the sufficient condition in Lemma \ref{le_sufficient}, a lower bound for $P\left( \left. \hat{x}_n\!\rightarrow\!x \right| \hat{x}_0\!\in\!\mathcal{B}\left(x\right) \right)$ is derived. }

By deriving the lower bound of probability for the sufficient condition $\underset{t\in I_m}{\sup} \left| \bar{x}(t) - \tilde{x}^{\tilde{n}(m)}(t)\right| \le \delta, \forall m \ge 0$ in Lemma \ref{le_sufficient}, we get the following lemma:
\begin{lemma}\label{le_lower_bound}
If (i) the initial point satisfies $\hat{x}_0 \in \mathcal{B}(x)$, (ii) $a_n$ is given by (\ref{eq_stepsize}) with any $\alpha > 0$,
then there exist  $N_0 \ge 0$ and $C_0>0$ such that\vspace{-1mm}
\begin{equation}\vspace{-1mm}\label{eq_lock}
P\left( \hat{x}_n \in \mathcal{I}, \forall n \ge 0 \right) \ge 1- 2e^{-C_0\cdot\frac{\rho}{\alpha^2}}.
\end{equation}
\end{lemma}
\begin{proof} 
\ifreport
See Appendix \ref{sec_proof_le_lower_bound}.
\else
See Appendix G in our technical report \cite{Li2017analog}. 
\fi
\end{proof}

Applying Lemma \ref{le_lower_bound} and Theorem 2.1 in \cite{borkar2002lock}, we can obtain
\begin{align}\label{eq_lock10}
	P\left( \left. \hat{x}_n \rightarrow x \right| \hat{x}_0 \in \mathcal{B} \right) \ge & P\left( \hat{x}_n \in \mathcal{I}, \forall n \ge 0 \right) 
	\ge 1 - 2e^{-C_0\cdot\frac{\rho}{\alpha^2}},
\end{align}
which completes the proof of Theorem \ref{th_lock}.
}

\section{Proof of Theorem~\ref{th_normal}}\label{proof_normal}

When the step-size $a_n$ is given by (\ref{eq_stepsize}) with any $\alpha > 0$ and $N_0 \ge 0$, Theorem 6.6.1 in \cite{nevel1973stochastic} has proposed the sufficient conditions to prove the asymptotic normality of $\hat{x}_n$, i.e., $\sqrt{n+N_0} \left( \hat{x}_n - x \right) \overset{d}{\rightarrow}\mathcal{N}\left( 0, \Sigma \right)$. Under the condition that $\hat{x}_n \rightarrow x$, we will prove that our algorithm satisfies its sufficient conditions and obtain the variance $\Sigma$ as follows:
\begin{itemize}
\item[1)] The estimate $\hat{x}_n$ should be within $[-1, 1]$.\vspace{0.1mm} \\
The projection operator in \eqref{eq_est} ensures that $\hat{x}_n \in [-1, 1]$.
\item[2)] Equation \eqref{eq_gx} should satisfy: (i) there exist an increasing sequence of $\sigma$-fields $\{\mathcal{F}_{n}: n \ge 0\}$ such that $\mathcal{F}_{m} \subset \mathcal{F}_{n}$ for $m < n$, and (ii) the random noise $\hat{z}_n$ is $\mathcal{F}_{n}$-measurable and independent of $\mathcal{F}_{n-1}$.\vspace{0.25mm} \\
As defined in Appendix \ref{proof_converge}, there exists an increasing sequence of $\sigma$-fields $\{ \mathcal{G}_n : n \ge 0 \}$, such that $\hat{z}_n$ is measurable w.r.t. $\mathcal{G}_{n}$, i.e., $\mathbb{E} \left[ \left. \hat{z}_n \right| \mathcal{G}_{n} \right] = \hat{z}_n$, and is independent of $\mathcal{G}_{n-1}$, i.e., $\mathbb{E} \left[ \left. \hat{z}_n \right| \mathcal{G}_{n-1}\right]  = \mathbb{E} \left[ \hat{z}_n \right] = 0$.
\item[3)] $\hat{x}_n$ should converge to $x$ almost surely as $n \rightarrow \infty$. \vspace{0.5mm} \\
Since $\hat{x}_n \rightarrow x$ is assumed, we have that $\hat{x}_n$ converges to $x$ almost surely as $n \rightarrow \infty$. \vspace{0.25mm}

\item[4)] {\blue The stable condition: \vspace{0.5mm}\\
From \eqref{eq_fx}, $f(v, x)$ can be rewritten as follows:\vspace{-1mm}
\begin{equation*}\vspace{-1mm}
\begin{aligned}
f(v, x) = & -\frac{\sin\left[ \frac{(M-1)\pi d}{\lambda} \left( v - x \right)\right]  \sin\left[ \frac{M\pi d}{\lambda} \left( v  - x \right)\right]}{\sqrt{M} \sin\left[ \frac{\pi d}{\lambda} \left(v - x \right) \right] } 
= c_1 \left( v - x \right) + o\left(v - x\right),
\end{aligned}
\end{equation*}
where $c_1$ is given by\vspace{-1mm}
\begin{equation*}\vspace{-1mm}
\begin{aligned}
c_1 = \left.\frac{\partial f(v, x)}{\partial v}\right|_{v = x} = - \frac{\sqrt{M}(M-1)\pi d}{\lambda}.
\end{aligned}
\end{equation*}
Then, we get the stable condition that\vspace{-1mm}
\begin{equation*}\vspace{-1mm}
A = c_1\alpha + \frac{1}{2} = -\frac{\sqrt{M}(M-1)\pi d \alpha}{\lambda}  + \frac{1}{2} < 0,
\end{equation*}
which results in $\alpha > \frac{\lambda}{2\sqrt{M}(M-1)\pi d}$.} \vspace{0.25mm}

\item[5)] The constraints for the random noise:\vspace{-1mm}
\begin{equation*}
\mathbb{E}\left[\left(\hat{z}_n\right)^2\right] = \frac{1}{2\rho} < \infty,
\end{equation*}
and\vspace{-1mm}
\begin{equation*}\vspace{-1mm}
\underset{V\rightarrow\infty}{\lim}\ \ \underset{n\ge1}{\sup}\ \ \int\limits_{\left| \hat{z}_n \right| > V} \left| \hat{z}_n \right|^2 p(\hat{z}_n) d\hat{z}_n = 0.
\end{equation*}

\end{itemize}
Hence, by Theorem 6.6.1 in \cite{nevel1973stochastic}, we have\vspace{-1mm}
\begin{equation*}\vspace{-1mm}\sqrt{n+N_0}\left( \hat{x}_n - x \right) \overset{d}{\rightarrow}\mathcal{N}\left( 0, \Sigma \right), \end{equation*}
where\vspace{-1mm}
\begin{equation}\vspace{-1mm}\label{eq_Sigma}\begin{aligned}
\Sigma = \alpha^2 \mathbb{E}\left[\left(\hat{z}_n\right)^2\right] \cdot \int_0^\infty e^{2Av} dv
= \frac{\alpha^2}{2\rho \left( \frac{2\sqrt{M}(M-1)\pi d \alpha}{\lambda} - 1 \right)}.\\ 
\end{aligned}\vspace{-2mm}\end{equation}
Due to that $\lim_{n\rightarrow\infty}\sqrt{{(n+N_0)}/{n}} = 1$, we have\vspace{-1mm}
\begin{equation*}\vspace{-1mm}
\sqrt{n}\left( \hat{x}_n - x \right) \rightarrow \sqrt{n}\cdot\sqrt{\frac{n+N_0}{n}}\left( \hat{x}_n - x \right) \overset{d}{\rightarrow}\mathcal{N}\left( 0, \Sigma \right),
\end{equation*}
as $n\rightarrow\infty$. By adapting $\alpha$ in \eqref{eq_Sigma}, we can obtain different $\Sigma$, which achieves the minimum value $\Sigma_{\min} = I_{\max}^{-1}$, i.e., the minimum CRLB in (\ref{eq_CRLB}), when $\alpha = \frac{\lambda}{\sqrt{M}(M-1)\pi d}$.

By assuming $\alpha = \frac{\lambda}{\sqrt{M}(M-1)\pi d}$, we can conclude that\vspace{-1mm}
\begin{equation*}\vspace{-1mm}\label{eq_normal2}
	\lim_{n\rightarrow\infty}~n~\mathbb{E}\left[\left(\hat{x}_n - x\right)^2\big| \hat{x}_n \rightarrow x\right] = I_{\max}^{-1}.
\end{equation*}

%
%

\ifreport
\section{Proof of Corollary~\ref{co_1}}\label{proof_co_1}

Let $\mathbf{u}(v) = [u_1(v)~ \cdots~ u_N(v)]^\text{T}$ be a $N$-dimensional vector function, which is first-order differentiable. Similar to (\ref{eq_MMSE})-(\ref{eq_CRLB}), its MSE is lower bounded by
\begin{align}
	\mathbb{E}\left[ \left\| \mathbf{u}(\hat{x}_{n}) - \mathbf{u}(x) \right\|^2_2 \right] = & \sum_{m=1}^N\mathbb{E}\left[ \left( u_m(\hat{x}_{n}) - u_m(x) \right)^2 \right] \nonumber \\ \label{eq_FI_co}
	\ge & \sum_{m=1}^N \frac{1}{nI'_{\max, m}},
\end{align}
where $I'_{\max,m}$ is given by
\begin{align}
I'_{\max,m} = &~\mathbb{E}\left[ \left. \left( \frac{\partial \log p \left( y_i|x, \mathbf{w}_i \right)}{\partial u_m(x)}\right)^2 \right| x, \mathbf{w}_i = \frac{\mathbf{a}(x)}{\sqrt{M}} \right].
\end{align}

According to Theorem~\ref{th_normal}, we have
\begin{equation}
	\lim_{n\rightarrow\infty}~n~\mathbb{E}\left[\left(\hat{x}_n - x\right)^2\big| \hat{x}_n \rightarrow x\right] = I_{\max}^{-1},
\end{equation}
where $I_{\max}$ is given by
\begin{equation}\begin{aligned}
I_{\max} = &~\mathbb{E}\left[ \left. - \frac{\partial^2 \log p \left( y_i|x, \mathbf{w}_i \right)}{\partial x^2} \right| x, \mathbf{w}_i = \frac{\mathbf{a}(x)}{\sqrt{M}} \right] \\
= &~\mathbb{E}\left[ \left. \left(\frac{\partial \log p \left( y_i|x, \mathbf{w}_i \right)}{\partial x} \right)^2 \right| x, \mathbf{w}_i = \frac{\mathbf{a}(x)}{\sqrt{M}} \right].
\end{aligned}\end{equation}
Since $\frac{\partial \log p \left( y_i|x, \mathbf{w}_i \right)}{\partial x}$ can be rewritten as
\begin{align}
\frac{\partial \log p \left( y_i|x, \mathbf{w}_i \right)}{\partial x} = &  \frac{\partial \log p \left( y_i|x, \mathbf{w}_i \right)}{\partial u_m(x)}\cdot u_m'(x),
\end{align}
we get
\begin{align}
I'_{\max,m} = &~\frac{I_{\max}}{\left[u_m'(x)\right]^2},
\end{align}
which results in
\begin{align}
\lim_{n\rightarrow\infty} n~\mathbb{E}\left[\left| u_m(\hat{x}_n) - u_m(x)\right|^2 \Big| \hat{x}_n \rightarrow x\right] = \left[ u_m'(x) \right]^2 I_{\max}^{-1}.\nonumber
\end{align}
Then, based on \eqref{eq_FI_co}, we conclude that
\begin{align}
\lim_{n\rightarrow\infty} n~\mathbb{E}\left[\left\| \mathbf{u}(\hat{x}_n) - \mathbf{u}(x)\right\|^2_2\Big| \hat{x}_n \rightarrow x\right] = \left\|{ {\mathbf{u}}'(x) }\right\|^2_2 I_{\max}^{-1}.\nonumber
\end{align}

\fi

%
%
%

\ifreport
\section{Proof of Lemma \ref{le_sufficient}}\label{sec_proof_le_sufficient}

When $m = 0$, $\tilde{x}^{\tilde{n}(0)}(T_0) = \bar{x}(T_0) = \hat{x}_0$. There are two symmetrical cases: (i) $\hat{x}_0 < x$ and (ii) $\hat{x}_0 > x$. We will consider the \emph{first case}, which can be directly extended to the \emph{second case}.

\emph{Case 1} ($\hat{x}_0 < x$): We will first prove that $\bar{x}(t) \in \mathcal{I} = \Big( x - |x - \hat{x}_0| - \delta,~x + |x - \hat{x}_0| + \delta \Big)$ for all $t \in I_0$.

If $\left| \bar{x}(t) - \tilde{x}^{\tilde{n}(0)}(t)\right| \le \delta$ for all $t \in I_0$, then we have $\bar{x}(t) - \tilde{x}^{\tilde{n}(0)}(t) \ge -\delta$. What's more, due to $\hat{x}_0 \in \mathcal{I} \subset \mathcal{B}(x)$ and the monotonic property of the ODE (\ref{eq_ODE}) within the mainlobe $\mathcal{B}(x)$, we get $\tilde{x}^{\tilde{n}(0)}(t) - \hat{x}_0 \ge 0$ and $x - \tilde{x}^{\tilde{n}(0)}(t) \ge 0$ for all $t \in I_0$. Therefore, we can obtain
\begin{align}\label{eq_left}
	&~\bar{x}(t) - (\hat{x}_0 - \delta) \\
	= &~\left[\bar{x}(t) - \tilde{x}^{\tilde{n}(0)}(t)\right] + \left[\tilde{x}^{\tilde{n}(0)}(t) - \hat{x}_0\right] + \delta \ge 0, \nonumber 
\end{align}
and
\begin{align}\label{eq_right}
	&~(x + |x - \hat{x}_0| + \delta) - \bar{x}(t) \\
	= &~(2x - \hat{x}_0 + \delta) - \bar{x}(t)
	= \left(x - \hat{x}_0\right) + \left[x - \bar{x}(t)\right] + \delta \nonumber\\
	= &~\left(x - \hat{x}_0\right) + \left[x - \tilde{x}^{\tilde{n}(0)}(t)\right] + \left[\tilde{x}^{\tilde{n}(0)}(t) - \bar{x}(t)\right] + \delta \ge 0,\nonumber
\end{align}
which result in $\bar{x}(t) \in \mathcal{I}$ for all $t \in I_0$.

Then, we consider the initial value $\bar{x}(T_1)$ for the next time interval $I_1$. With the $T$ given by \eqref{eq_T}, we have
\begin{align*}
x - \hat{x}_0 \ge \tilde{x}^{\tilde{n}(0)}(T_1) - \hat{x}_0 \ge \tilde{x}^{\tilde{n}(0)}(T) - \hat{x}_0 > \delta.
\end{align*}
Therefore, we get
\begin{align}\label{eq_left2}
	&~\bar{x}(T_1) - \hat{x}_0 \\
	= &~\left[\bar{x}(T_1) - \tilde{x}^{\tilde{n}(0)}(T_1)\right] + \left[\tilde{x}^{\tilde{n}(0)}(T_1) - \hat{x}_0\right] \ge 0,\nonumber
\end{align}
and
\begin{align}\label{eq_right2}
	&~(x + |x - \hat{x}_0|) - \bar{x}(T_1) \\
	= &~(2x - \hat{x}_0) - \bar{x}(T_1)	= \left(x - \hat{x}_0\right) + \left[x - \bar{x}(T_1)\right] \nonumber\\
	= &~\left(x - \hat{x}_0\right) + \left[x - \tilde{x}^{\tilde{n}(0)}(T_1)\right] + \left[\tilde{x}^{\tilde{n}(0)}(T_1) - \bar{x}(T_1)\right] 	\ge 0,\nonumber
\end{align}
which result in $\bar{x}(T_1) \in \big[ x - |x - \hat{x}_0|,~x + |x - \hat{x}_0| \big]$.

\emph{Case 2} ($\hat{x}_0 > x$): Owing to symmetric property, we can use the same method as \eqref{eq_left}-\eqref{eq_right2} to obtain that $\bar{x}(t) \in \mathcal{I}$ for all $t \in I_0$ and $\bar{x}(T_1) \in \big[ x - |x - \hat{x}_0|,~x + |x - \hat{x}_0| \big]$.

\vspace{3mm}
When $m = 1$, $\tilde{x}^{\tilde{n}(1)}(T_1) = \bar{x}(T_1) \in \big[ x - |x - \hat{x}_0|,~x + |x - \hat{x}_0| \big]$. If $\bar{x}(T_1) < x$ and $\left| \bar{x}(t) - \tilde{x}^{\tilde{n}(1)}(t)\right| \le \delta$, then for all $t \in I_1$, we have $\bar{x}(T_1) \ge  \hat{x}_0$, $\tilde{x}^{\tilde{n}(1)}(t) - \hat{x}_0 \ge 0$, $x - \tilde{x}^{\tilde{n}(1)}(t) \ge 0$, and
\begin{align*}
x - \hat{x}_0 \ge \tilde{x}^{\tilde{n}(1)}(T_2) - \hat{x}_0 \ge \tilde{x}^{\tilde{n}(1)}(T_1+T) - \hat{x}_0 > \delta.
\end{align*}
Similar to \eqref{eq_left}-\eqref{eq_right2}, we can get $\bar{x}(t) \in \mathcal{I}~\text{for all}~t \in I_1$ and $\bar{x}(T_2) \in \big[ x - |x - \hat{x}_0|,~x + |x - \hat{x}_0| \big]$, which are also true for the case that $\bar{x}(T_1) > x$.

Hence, we can use the same method to prove the case of $m \ge 2$, which finally yields $\bar{x}(t) \in \mathcal{I}$ for all $t \in I_m$ and $m \ge 0$. Since $\bar{x}(t_n) = \hat{x}_n$ for all $n \ge 0$, we can obtain that $\hat{x}_n \in \mathcal{I}~\text{for all}~n \ge 0$, which completes the proof.

\vspace{-3mm}
\section{Proof of Lemma \ref{le_lower_bound}}\label{sec_proof_le_lower_bound}
\vspace{-1mm}

The following lemmas are needed to prove Lemma \ref{le_lower_bound}:
\begin{lemma}\vspace{-1mm}\label{le_gronwall}
Let $n_T \overset{\Delta}{=} \inf \left\{i \in \mathbb{Z}: t_{n+i} \ge t_n + T \right\}$. If there exists a constant $C>0$, which satisfies\vspace{-1mm}
\begin{equation}\vspace{-1mm}\label{eq_gronwall1}
\begin{aligned}
	&~\left| \bar{x}(t_{n+m}) - \tilde{x}^n(t_{n+m})\right| \\
	\le &~L \sum_{i=1}^{m} a_{n+i} \left| \bar{x}(t_{n+i-1}) - \tilde{x}^n(t_{n+i-1}) \right|  + C,
\end{aligned}
\end{equation}
for all $n \ge 0$ and $1 \le m \le n_T$, then\vspace{-1mm}
\begin{equation}\vspace{-1mm}\label{eq_gronwall2}
\begin{aligned}
	\underset{t\in\left[ t_n, t_{n+n_T} \right]}{\sup} \left| \bar{x}(t) - \tilde{x}^n(t)\right| \le  \frac{\sqrt{M} a_{n+1}}{2} + C e^{L (T+a_1)}.
\end{aligned}
\end{equation}
\end{lemma}
\begin{proof}
See Appendix \ref{sec_proof_le_gronwall}.
\end{proof}

\begin{lemma}\vspace{-1mm}\label{le_cheb}
If $\{M_i: i = 1, 2, \ldots\}$ satisfies that: (i)  $M_i$ is Gaussian distributed with zero mean, and (ii) $M_i$ is a martingale in $i$, then\vspace{-1mm}
\begin{equation}\vspace{-1mm}\label{eq_lock5}
\begin{aligned}
	& P\left( \underset{0\le i \le k}{\sup}\left|M_i\right| > \eta \right) \le 2\exp\left\{-\frac{\eta^2}{2\operatorname{Var}\left[M_k\right]}\right\},
\end{aligned}
\end{equation}
for any $\eta > 0$.
\end{lemma}
\begin{proof}
See Appendix \ref{sec_proof_le_cheb}.
\end{proof}

\begin{lemma}\vspace{-1mm}\label{le_sum}
If given a constant $C > 0$, then\vspace{-1mm}
\begin{equation}\vspace{-1mm}\label{eq_increasing}
\begin{aligned}
	G(v) = \frac{1}{v}\exp\left[-\frac{C}{v}\right],
\end{aligned}
\end{equation}
is increasing for all $0 < v < C$.
\end{lemma}
\begin{proof}
The derivative of $G(v)$ is\vspace{-1mm}
\begin{equation*}\vspace{-1mm}
	G'(v) = \frac{C - v}{v^3}\exp\left[-\frac{C}{v}\right].
\end{equation*}

Let $G'(v) > 0$ and we can obtain that $G(v)$ is increasing for $v \in (0, C)$, which completes the proof.
\end{proof}

{\blue A lower bound of the probability that the sequence $\{\hat{x}_n\}$ remains in the invariant set $\mathcal{I}$ can be derived as follows:\vspace{-0.5mm}
\begin{equation}\vspace{-1mm}\label{eq_p_invariant}
\begin{aligned}
	& P\left( \hat{x}_n \in \mathcal{I}, \forall n \ge 0 \right) \\
	\overset{(a)}{\ge} & P\left( \underset{t\in I_m}{\sup} \left| \bar{x}(t) - \tilde{x}^{\tilde{n}(m)}(t)\right| \le \delta, \forall m \ge 0 \right)  \\
	\overset{(b)}{\ge} & 1 - \sum_{m\ge 0} P\left( \left. \underset{t\in I_m}{\sup} \left| \bar{x}(t) - \tilde{x}^{\tilde{n}(m)}(t)\right| > \delta \right| \right.  \\
	&~~~~~~~~~~~~~~~\left. \underset{t\in I_i}{\sup} \left| \bar{x}(t) - \tilde{x}^{\tilde{n}(i)}(t)\right| \le \delta, 0 \le i < m \right)
\end{aligned}
\end{equation}
where Step $(a)$ is due to Lemma \ref{le_sufficient}, Step $(b)$ is due to Lemma 4.2 in \cite{borkar2008stochastic}. To calculate the lower bound of probability, we denote the continuous processes in \eqref{eq_continuous} and \eqref{eq_ODE_new} as follows:\vspace{-1mm}
\begin{equation}\vspace{-1mm}\label{eq_seq_trace}
\begin{aligned} \bar{x}(t_{n+m}) = &~\bar{x}(t_n)  + \sum_{i=1}^{m} a_{n+i} f(\bar{x}(t_{n+i-1}), x) \\
&~+ (\xi_{n+m} - \xi_{n}),
\end{aligned}\end{equation}
and\vspace{-1mm}
\begin{equation}\vspace{-1mm}\label{eq_ode_trace}\begin{aligned}
\tilde{x}^n(t_{n+m}) = &~\tilde{x}^n(t_n) + \int_{t_n}^{t_{n+m}} f(\tilde{x}^n(v), x) dv \\
= &~\tilde{x}^n(t_n) + \sum_{i=1}^{m} a_{n+i} f(\tilde{x}^n(t_{n+i-1}), x) \\
	&~\!\!\!\!\!+ \int_{t_n}^{t_{n+m}} \left[f(\tilde{x}^n(v), x) - f(\tilde{x}^n(\underline{v}), x) \right]dv,
\end{aligned}\end{equation}
for $n \ge 0, 1 \le m \le n_T$, where $\xi_0 \overset{\Delta}{=} 0, \xi_n \overset{\Delta}{=} \sum_{i=1}^{n} a_{i} \hat{z}_{i} $, $n \ge 1$, $\hat{z}_i = -{\operatorname{Im}\left\{{z}_i\right\}}/{\sqrt{\rho}} \sim \mathcal{N}\left( 0, {1}/({2\rho})\right)$, and $\underline{v} \overset{\Delta}{=} \max \left\{ t_n: t_n \le v, n \ge 0 \right\}$ for $v \ge 0$.

In order to bound $\left| \bar{x}(t_{n+m}) - \tilde{x}^n(t_{n+m})\right|$, we first obtain the Lipschitz constant of the function $f(v,x)$ w.r.t. the first variable $v$, given by\vspace{-0.5mm}
\begin{equation}\vspace{-0.5mm}\label{eq_Lip}
	L \overset{\Delta}{=} \underset{v_1 \ne v_2}{\sup} \frac{\left| f(v_1,x) - f(v_2,x) \right|}{\left| v_1 - v_2 \right|} = \frac{\sqrt{M}(M-1)\pi d}{\lambda}.
\end{equation}
In addition, similar to \eqref{eq_expectation_yn}, we can obtain}\vspace{-0.5mm}
\begin{equation}\vspace{-0.5mm}\label{eq_CT}
\begin{aligned}
	\left| f(\tilde{x}^n(t), x) \right| \le \sqrt{M}, \forall t \ge t_n.
\end{aligned}
\end{equation}
Hence, we have\vspace{-0.5mm}
\begin{equation}\vspace{-0.5mm}\label{eq_int}
\begin{aligned}
& \left| \int_{t_n}^{t_{n+m}} \left[f(\tilde{x}^n(v), x) - f(\tilde{x}^n(\underline{v}, x)) \right]dv \right| \\
\le & \int_{t_n}^{t_{n+m}} \left| f(\tilde{x}^n(v), x) - f(\tilde{x}^n(\underline{v}), x) \right| dv \\
\overset{(a)}{\le} & \int_{t_n}^{t_{n+m}} L \left| \tilde{x}^n(v) - \tilde{x}^n(\underline{v}) \right| dv \\
\overset{(b)}{\le} & \int_{t_n}^{t_{n+m}} L \left| \int_{\underline{v}}^{v} f(\tilde{x}^n(s), x) ds \right| dv \\
\le & \int_{t_n}^{t_{n+m}} \int_{\underline{v}}^{v} L \left| f(\tilde{x}^n(s), x) \right| ds dv \\
\overset{(c)}{\le} & \int_{t_n}^{t_{n+m}} \int_{\underline{v}}^{v} \sqrt{M} L ds dv =   \int_{t_n}^{t_{n+m}}\sqrt{M} L (v - \underline{v}) dv \\
=  & \sum_{i=1}^{m} \int_{t_{n+i-1}}^{t_{n+i}}\sqrt{M} L (v - t_{n+i-1}) dv \\
= & \sum_{i=1}^{m} \frac{\sqrt{M} L (t_{n+i} - t_{n+i-1})^2}{2}= \frac{\sqrt{M} L}{2} \sum_{i=1}^{m} a_{n+i}^2,
\end{aligned}
\end{equation}
where Step $(a)$ uses (\ref{eq_Lip}), Step $(b)$ is due to the definition in (\ref{eq_ODE_new}), and Step $(c)$ uses (\ref{eq_CT}). {\blue Then, by using (\ref{eq_seq_trace})-(\ref{eq_Lip}) and (\ref{eq_int}), we can get\vspace{-0.5mm}
\begin{equation}\vspace{-0.5mm}\label{eq_lock2}
\begin{aligned}
	& \left| \bar{x}(t_{n+m}) - \tilde{x}^n(t_{n+m})\right| \\
	\le & L \sum_{i=1}^{m} a_{n+i} \left| \bar{x}(t_{n+i-1}) - \tilde{x}^n(t_{n+i-1}) \right| \\
	&  + \frac{\sqrt{M} L}{2} \sum_{i=1}^{n_T} a_{n+i}^2+ \underset{1 \le m\le n_T}{\sup}|\xi_{n+m} - \xi_{n}|,
\end{aligned}
\end{equation}
for $n \ge 0, 1 \le m \le n_T$. Applying Lemma \ref{le_gronwall} to (\ref{eq_lock2}) and letting $C = \frac{\sqrt{M} L}{2} \sum_{i=1}^{n_T} a_{n+i}^2+ \underset{1\le m\le n_T}{\sup}|\xi_{n+m} - \xi_{n}|, n = \tilde{n}(m)$,
yields\vspace{-1mm}
\begin{equation}\vspace{-1mm}\label{eq_lock3-2}
\begin{aligned}
	&~\underset{t\in I_m}{\sup} \left| \bar{x}(t) - \tilde{x}^{\tilde{n}(m)}(t)\right|	 \\
	\le &~C_e \left\{ \frac{\sqrt{M} L}{2} \big[ b(\tilde{n}(m)) - b(\tilde{n}(m+1)) \big]\right. \\
	& \left.  + \underset{\tilde{n}(m) \le k \le \tilde{n}(m+1)}{\sup}\left|\xi_{k} - \xi_{\tilde{n}(m)}\right| \right\} + \frac{\sqrt{M} a_{\tilde{n}(m)+1}}{2},
\end{aligned}
\end{equation}
where $\tilde{n}(m+1) = n + n_T$, $C_e \overset{\Delta}{=} e^{L (T+a_1)}$, and $b(n) \overset{\Delta}{=} \sum_{i > n} a_{i}^2$. Also, we suppose that the step-sizes $\{a_n\}$ satisfy\vspace{-1mm}
\begin{equation}\vspace{-1mm}\label{eq_lock_constr}
C_e \frac{\sqrt{M} L}{2} \big[b(\tilde{n}(m)) - b(\tilde{n}(m+1))\big] +  \frac{\sqrt{M} a_{\tilde{n}(m)+1}}{2} < \frac{\delta}{2},
\end{equation}
for $m \ge 0$. Then, given $\underset{t\in I_m}{\sup} \left| \bar{x}(t) - \tilde{x}^{\tilde{n}(m)}(t)\right| > \delta$, we can obtain from \eqref{eq_lock3-2} and \eqref{eq_lock_constr} that\vspace{-1mm}
\begin{equation*}\vspace{-1mm}
\begin{aligned}
	&~\underset{\tilde{n}(m)\le k \le \tilde{n}(m+1)}{\sup}\left|\xi_{k} - \xi_{\tilde{n}(m)}\right| \\
		> &~\frac{1}{C_e}\left( \underset{t\in I_m}{\sup} \left| \bar{x}(t) - \tilde{x}^{\tilde{n}(m)}(t)\right| - \frac{\delta}{2} \right) > \frac{\delta}{2C_e}.
\end{aligned}
\end{equation*}
Hence, we get}\vspace{-1mm}
\begin{equation}\vspace{-1mm}\label{eq_lock4}
\begin{aligned}
	&~P\left( \left. \underset{t\in I_m}{\sup} \left| \bar{x}(t) - \tilde{x}^{\tilde{n}(m)}(t)\right| > \delta \right|\right. \\
	&~~~~~~\left.\underset{t\in I_i}{\sup} \left| \bar{x}(t) - \tilde{x}^{\tilde{n}(i)}(t)\right| \le \delta, 0 \le i < m \right) \\
	{\le} & P\left( \left. \underset{\tilde{n}(m)\le k \le \tilde{n}(m+1)}{\sup}\left|\xi_{k} - \xi_{\tilde{n}(m)}\right| > \frac{\delta}{2C_e} \right| \right. \\
	&~~~~~~\left.  \underset{t\in I_i}{\sup} \left| \bar{x}(t) - \tilde{x}^{\tilde{n}(i)}(t)\right| \le \delta, 0 \le i < m \right) \\
	\overset{(a)}{=} &~P\left( \underset{\tilde{n}(m)\le k \le \tilde{n}(m+1)}{\sup}\left|\xi_{k} - \xi_{\tilde{n}(m)}\right| > \frac{\delta}{2C_e} \right),
\end{aligned}
\end{equation}
where Step $(a)$ is due to the independence of noise, i.e., $\left( \xi_{k} - \xi_{\tilde{n}(m)} \right), \tilde{n}(m) \le k \le \tilde{n}(m+1)$ are independent of $\hat{x}_n, 0 \le n \le \tilde{n}(m)$.

%


{\blue Next, we will consider how to calculate a lower bound for \eqref{eq_lock4}. With the increasing $\sigma$-fields $\{\!\mathcal{G}_n\!:\!n\!\ge\!0\!\}$ defined in Appendix \ref{proof_converge}, we have for $n \ge 0$,
\begin{itemize}
\item[1)] $\xi_n = \sum_{m=1}^{n} a_{m} \hat{z}_{m} \sim \mathcal{N}(0, \sum_{m=1}^n \frac{a_m^2}{2\rho})$,

\item[2)] $\xi_n$ is $\mathcal{G}_n$-measurable, i.e., $\mathbb{E} \left[ \left. \xi_n \right| \mathcal{G}_n \right] = \xi_n$,

\item[3)] $\mathbb{E} \left[ \left| \xi_n \right|^2 \right] = \sum_{m=1}^n \frac{a_m^2}{2\rho} < \infty$,

\item[4)] $\mathbb{E} \left[ \left. \xi_n \right| \mathcal{G}_m \right] = \xi_m$ for all $0 \le m < n$.
\end{itemize}
Therefore, $\xi_n$ is Gaussian distributed with zero mean, and is a martingale w.r.t. $\mathcal{G}_n$. Letting $\eta = \frac{\delta}{2C_e}$, $M_i = \xi_{\tilde{n}(m)+i} - \xi_{\tilde{n}(m)}$ and $k = {\tilde{n}(m+1) - \tilde{n}(m)}$ in Lemma \ref{le_cheb}, we can obtain\vspace{-1mm}
\begin{equation}\vspace{-1mm}\label{eq_lock6}
\begin{aligned}
	&~ P\left( \underset{\tilde{n}(m)\le k \le \tilde{n}(m+1)}{\sup}\left|\xi_{k} - \xi_{\tilde{n}(m)}\right| > \frac{\delta}{2C_e} \right) \\
	\le & ~2\exp\left\{-\frac{\delta^2}{8C_e^2\operatorname{Var}\left[\xi_{\tilde{n}(m+1)} - \xi_{\tilde{n}(m)}\right]}\right\} \\
	= & ~2\exp\left\{-\frac{\rho\delta^2}{4C_e^2\big[b(\tilde{n}(m)) - b(\tilde{n}(m+1))\big]}\right\}.
\end{aligned}
\end{equation}

Combining \eqref{eq_p_invariant}, \eqref{eq_lock4} and \eqref{eq_lock6}, we have\vspace{-1mm}
\begin{equation}\vspace{-1mm}\label{eq_lock7}
\begin{aligned}
\!&~P\left( \hat{x}_n \in \mathcal{I}, \forall n \ge 0 \right)  \\
	\!\ge &~1 -  2\sum_{m \ge 0} \exp\left\{-\frac{\rho\delta^2}{4C_e^2\big[b(\tilde{n}(m)) - b(\tilde{n}(m+1))\big]}\right\}. 
\end{aligned}
\end{equation}
To further simplify \eqref{eq_lock7}, we assume that the step-sizes satisfy\vspace{-1mm}
\begin{equation}\vspace{-1mm}\label{eq_lock_constr2}
b(0) = \sum_{i > 0} a_i^2 \le \frac{\rho\delta^2}{4C_e^2}.
\end{equation}
Then, from Lemma \ref{le_sum}, we can obtain\vspace{-1mm}
\begin{equation}\vspace{-1mm}\label{eq_lock7-2}
\begin{aligned}
	&~\sum_{m\ge 0} \exp\left\{-\frac{\rho\delta^2}{4C_e^2\big[b(\tilde{n}(m)) - b(\tilde{n}(m+1))\big]}\right\} \\
	\le &~\sum_{m \ge 0} \left[ b(\tilde{n}(m)) - b(\tilde{n}(m+1))\right] \cdot \frac{\exp\left\{-\frac{\rho\delta^2}{4C_e^2b(0)}\right\} }{b(0)} \\
	= &~b(0) \cdot \frac{\exp\left\{-\frac{\rho\delta^2}{4C_e^2b(0)}\right\} }{b(0)} = \exp\left\{-\frac{\rho\delta^2}{4C_e^2b(0)}\right\}.
\end{aligned}
\end{equation}
where $b(\tilde{n}(m)) - b(\tilde{n}(m+1)) \le b(0)$. As $C_e = e^{L (T+a_1)}$, $b(0) = \sum_{i > 0} a_{i}^2$, and $a_n, T, L$ are given by \eqref{eq_stepsize}, \eqref{eq_T}, \eqref{eq_Lip}  separately, we can obtain\vspace{-1mm}
\begin{equation}\vspace{-1mm}\label{eq_exponential}
\frac{\rho\delta^2}{4C_e^2b(0)} = \frac{\delta^2}{4 e^{2L (T+\frac{\alpha}{N_0+1})} \sum_{i \ge 1} \frac{1}{(i+N_0)^2}} \cdot \frac{\rho}{\alpha^2}.
\end{equation}
To ensures that $\hat{x}_0 + a_{1}f(\hat{x}_0,x)$ does not exceed the mainlobe $\mathcal{B}(x)$, i.e., the first step-size $a_{1}$ satisfies\vspace{-1mm}
\begin{equation*}\vspace{-1mm}\left|\hat{x}_0 + a_{1}f(\hat{x}_0,x) - x\right| < \frac{\lambda}{Md},\end{equation*}
we can obtain the maximum $\alpha$ as follows:\vspace{-1mm}
\begin{equation*}\vspace{-1mm}
	\alpha_{\max} = \frac{(N_0+1)\left(\left|x-\hat{x}_0\right|+\frac{\lambda}{Md}\right)}{\left|f(\hat{x}_0, x)\right|}.
\end{equation*}
Hence, from \eqref{eq_exponential}, we have\vspace{-1mm}
\begin{equation}\vspace{-1mm}\label{eq_exponential_new}
\frac{\rho\delta^2}{4C_e^2b(0)} \!\cdot\! \frac{\alpha^2}{\rho}  \!\ge\!  \frac{\delta^2}{4 e^{2L (T+\frac{\alpha_{\max}}{N_0+1})} \sum_{i \ge 1} \frac{1}{(i+N_0)^2}} \overset{\Delta}{=} C_0 > 0,
\end{equation}
where \eqref{eq_delta_choice}, (\ref{eq_lock_constr}) and (\ref{eq_lock_constr2}) should be satisfied, which can be achieved with a sufficiently large $N_0 \ge 0$. }

Finally, from \eqref{eq_lock7}, \eqref{eq_lock7-2} and \eqref{eq_exponential_new}, we can get\vspace{-1mm}
\begin{equation*}\vspace{-1mm}
\begin{aligned}
P\left( \hat{x}_n \in \mathcal{I}, \forall n \ge 0 \right) \ge 1 -  2e^{-C_0\cdot\frac{\rho}{\alpha^2}},
\end{aligned}
\end{equation*}
which completes the proof.

\section{Proof of Lemma \ref{le_gronwall}}\label{sec_proof_le_gronwall}
Apply the discrete Gronwall inequality \cite{holte2009discrete}, leading (\ref{eq_gronwall1}) to
\begin{equation}\label{eq_gronwall3}
\begin{aligned}
	\left| \bar{x}(t_{n+m}) - \tilde{x}^n(t_{n+m})\right| \le C e^{L\sum_{i=1}^m a_{n+i}}.
\end{aligned}
\end{equation}
Since $1 \le m \le n_T$ and $n_T = \inf \left\{i \in \mathbb{Z}: t_{n+i} \ge t_n + T \right\} $, we get
\begin{equation}\label{eq_gronwall33}
\begin{aligned}
	\sum_{i=1}^m a_{n+i} = t_{n+m} - t_n \le T + a_{n+n_T} \le T + a_1.
\end{aligned}
\end{equation}
By combining \eqref{eq_gronwall3} and \eqref{eq_gronwall33}, we have
\begin{equation}\label{eq_gronwall4}
\begin{aligned}
	\left| \bar{x}(t_{n+m}) - \tilde{x}^n(t_{n+m})\right| \le &~C e^{L (T+a_1)}.
\end{aligned}
\end{equation}

For $\forall t \in [t_{n+m-1}, t_{n+m}], 1 \le m \le n_T$, from \eqref{eq_continuous}, we have
\begin{equation*}\label{eq_gronwall5}
\begin{aligned}
	\bar{x}(t) & = \bar{x}(t_{n+m-1}) + \frac{(t-t_{n+m-1})\left[\bar{x}(t_{n+m}) - \bar{x}(t_{n+m-1})\right]}{a_{n+m}}\\
	& = \gamma \bar{x}(t_{n+m-1}) + (1 - \gamma) \bar{x}(t_{n+m}),
\end{aligned}
\end{equation*}
where $\gamma = \frac{t_{n+m} - t}{a_{n+m}} \in [0, 1]$. Then, we can get \eqref{eq_gronwall6} on the top of the next page, where Step $(a)$ is according to the definition of $\tilde{x}^n(t)$ in \eqref{eq_ODE_new}, Step $(b)$ is due to (\ref{eq_gronwall4}), Step $(c)$ is obtained from (\ref{eq_CT}), and Step $(d)$ is obtained by using $\gamma = \frac{t_{n+m} - t}{a_{n+m}}$.

\begin{figure*}[t]
\begin{equation}\begin{aligned}\label{eq_gronwall6}
	&~\left| \bar{x}(t) - \tilde{x}^n(t)\right| \\
	= &~\left| \gamma(\bar{x}(t_{n+m-1}) - \tilde{x}^n(t)) + (1 - \gamma) (\bar{x}(t_{n+m}) - \tilde{x}^n(t)) \right| \\
	\overset{(a)}{=} &~\left| \gamma \left[ \bar{x}(t_{n+m-1}) - \tilde{x}^n(t_{n+m-1}) - \int_{t_{n+m-1}}^t f(\tilde{x}^n(s), x) ds \right] \right. \\&~\left.+ (1 - \gamma) \left[ \bar{x}(t_{n+m}) - \tilde{x}^n(t_{n+m}) -  \int_{t_{n+m}}^t f(\tilde{x}^n(s), x) ds \right] \right| \\
	\le &~\gamma\left| \int_{t_{n+m-1}}^t f(\tilde{x}^n(s), x) ds \right| + (1 - \gamma)\left| \int_{t_{n+m}}^t f(\tilde{x}^n(s), x) ds \right| \\
	&~+ \gamma\left| \bar{x}(t_{n+m-1}) - \tilde{x}^n(t_{n+m-1}) \right| + (1 - \gamma)\left| \bar{x}(t_{n+m}) - \tilde{x}^n(t_{n+m}) \right| \\
	\overset{(b)}{\le} &~\gamma  \int_{t_{n+m-1}}^t \left|f(\tilde{x}^n(s), x) \right| ds + (1 - \gamma)\int_t^{t_{n+m}} \left| f(\tilde{x}^n(s), x) \right| ds + Ce^{L (T+a_1)} \\
	\overset{(c)}{\le} &~\sqrt{M} \gamma (t - t_{n+m-1}) + \sqrt{M}(1 - \gamma)(t_{n+m} - t) + Ce^{L (T+a_1)} \\
	\overset{(d)}{\le} &~2\sqrt{M} a_{n+m}\gamma(1-\gamma)  + Ce^{L (T+a_1)} \le \frac{\sqrt{M} a_{n+m}}{2} + Ce^{L (T+a_1)} \\
	\le &~\underset{1 \le m \le n_T}{\sup} \frac{\sqrt{M} a_{n+m}}{2} + C e^{L (T+a_1)} = \frac{\sqrt{M} a_{n+1}}{2} + C e^{L (T+a_1)}.
\end{aligned}\end{equation}
\hrulefill
\end{figure*}

Therefore, from \eqref{eq_gronwall6}, we can obtain
\begin{equation*}
\begin{aligned}
	& \underset{t\in\left[ t_n, t_{n+n_T} \right]}{\sup} \left| \bar{x}(t) - \tilde{x}^n(t)\right| 	\le \frac{\sqrt{M} a_{n+1}}{2} + C e^{L (T+a_1)},
\end{aligned}
\end{equation*}
which completes the proof.


\section{Proof of Lemma \ref{le_cheb}}\label{sec_proof_le_cheb}
As $M_i$ is Gaussian distributed with zero mean, and is a martingale in $i$.
By utilizing the Doob's inequality \cite{Hoeffding1963Probability} for $\eta > 0$, we have
\begin{equation}
\begin{aligned}
	& P\left( \underset{0\le i \le k}{\sup}M_i > \eta \right) \le & \frac{\mathbb{E}\left[ e^{CM_k} \right]}{e^{C\eta}}.
\end{aligned}
\end{equation}

Due to the property of zero-mean Gaussian distribution, we have
\begin{equation}
	\mathbb{E}\left[ e^{CM_k} \right] = \exp\left\{\frac{C^2}{2}\operatorname{Var}\left[M_k\right]\right\}.
\end{equation}
Then we can obtain
\begin{equation}\label{eq_doob}
\begin{aligned}
	& P\left( \underset{0\le i \le k}{\sup}M_i > \eta \right) \le & \exp\left\{\frac{C^2}{2}\operatorname{Var}\left[M_k\right] - C\eta\right\}.
\end{aligned}
\end{equation}
We choose the $C$ to minimize the upper bound above, which yields $C = \frac{\eta}{\operatorname{Var}\left[M_k\right]}$. Therefore, we have
\begin{equation}
	P\left( \underset{0\le i \le k}{\sup}M_i > \eta \right) \le \exp\left\{-\frac{\eta^2}{2\operatorname{Var}\left[M_k\right]}\right\}.
\end{equation}

Because the distribution of $\left\{ M_1, M_2, \ldots, M_k \right\}$ is symmetric, we get
\begin{equation}
\begin{aligned}
	& P\left( \underset{0\le i \le k}{\sup}\left|M_i\right| > \eta \right) \\
	= & P\left( \underset{0\le i \le k}{\sup}M_i > \eta \bigcup \underset{0\le i \le k}{\inf}M_i < - \eta \right)\\
	\le & P\left( \underset{0\le i \le k}{\sup}M_i > \eta\right) + P\left( \underset{0\le i \le k}{\inf}M_i < - \eta\right) \\
	= & 2P\left( \underset{0\le i \le k}{\sup}M_i > \eta\right).
\end{aligned}
\end{equation}
Hence, we have
\begin{equation*}
\begin{aligned}
	& P\left( \underset{0\le i \le k}{\sup}\left|M_i\right| > \eta \right) \le 2\exp\left\{-\frac{\eta^2}{2\operatorname{Var}\left[M_k\right]}\right\},
\end{aligned}
\end{equation*}
which completes the proof.

\fi
\fi

\bibliographystyle{IEEEtran}
\bibliography{IEEEabrv,reference}

\begin{thebibliography}{10}
\providecommand{\url}[1]{#1}
\csname url@samestyle\endcsname
\providecommand{\newblock}{\relax}
\providecommand{\bibinfo}[2]{#2}
\providecommand{\BIBentrySTDinterwordspacing}{\spaceskip=0pt\relax}
\providecommand{\BIBentryALTinterwordstretchfactor}{4}
\providecommand{\BIBentryALTinterwordspacing}{\spaceskip=\fontdimen2\font plus
\BIBentryALTinterwordstretchfactor\fontdimen3\font minus
  \fontdimen4\font\relax}
\providecommand{\BIBforeignlanguage}[2]{{%
\expandafter\ifx\csname l@#1\endcsname\relax
\typeout{** WARNING: IEEEtran.bst: No hyphenation pattern has been}%
\typeout{** loaded for the language `#1'. Using the pattern for}%
\typeout{** the default language instead.}%
\else
\language=\csname l@#1\endcsname
\fi
#2}}
\providecommand{\BIBdecl}{\relax}
\BIBdecl

\bibitem{li2017conf}
J.~Li, Y.~Sun, L.~Xiao, S.~Zhou, and C.~E. Koksal, ``Analog beam tracking in
  linear antenna arrays: {C}onvergence and optimality,'' in \emph{51st Asilomar
  Conference on Signals, Systems, and Computers}, Oct. 2017.

\bibitem{SunVR2018}
Y.~{Sun}, Z.~{Chen}, M.~{Tao}, and H.~{Liu}, ``Communication, computing and
  caching for mobile {VR} delivery: Modeling and trade-off,'' in \emph{IEEE
  ICC}, May 2018, pp. 1--6.

\bibitem{Wang2009V2V}
C.~X. Wang, X.~Cheng, and D.~I. Laurenson, ``Vehicle-to-vehicle channel
  modeling and measurements: recent advances and future challenges,''
  \emph{IEEE Commun. Mag.}, vol.~47, no.~11, Nov. 2009.

\bibitem{Wu2016Survey}
J.~Wu and P.~Fan, ``A survey on high mobility wireless communications:
  Challenges, opportunities and solutions,'' \emph{IEEE Access}, vol.~4, Jan.
  2016.

\bibitem{Xiao2016Enabling}
Z.~Xiao, P.~Xia, and X.~G. Xia, ``Enabling {UAV} cellular with millimeter-wave
  communication: potentials and approaches,'' \emph{IEEE Commun. Mag.},
  vol.~54, no.~5, May 2016.

\bibitem{Pi2011An}
Z.~Pi and F.~Khan, ``An introduction to millimeter-wave mobile broadband
  systems,'' \emph{IEEE Commun. Mag.}, vol.~49, no.~6, Jun. 2011.

\bibitem{Boccardi2014Five}
F.~Boccardi, R.~W. Heath, A.~Lozano, T.~L. Marzetta, and P.~Popovski, ``Five
  disruptive technology directions for 5{G},'' \emph{IEEE Commun. Mag.},
  vol.~52, no.~2, Feb. 2014.

\bibitem{Heath2016overview}
R.~W. Heath, N.~Gonz\'alez-Prelcic, S.~Rangan, W.~Roh, and A.~M. Sayeed, ``An
  overview of signal processing techniques for millimeter wave {MIMO}
  systems,'' \emph{IEEE J. Sel. Top. Signal Process.}, Apr. 2016.

\bibitem{Rappaport2013Millimeter}
T.~S. Rappaport, S.~Sun, R.~Mayzus, H.~Zhao, Y.~Azar, K.~Wang, G.~N. Wong,
  J.~K. Schulz, M.~Samimi, and F.~Gutierrez, ``Millimeter wave mobile
  communications for 5{G} cellular: {i}t will work!'' \emph{IEEE Access},
  vol.~1, May 2013.

\bibitem{Rappaport2015Wideband}
T.~S. Rappaport, G.~R. MacCartney, M.~K. Samimi, and S.~Sun, ``Wideband
  millimeter-wave propagation measurements and channel models for future
  wireless communication system design,'' \emph{IEEE Trans. Commun.}, vol.~63,
  no.~9, Sep. 2015.

\bibitem{Hang2010BeamAdapt}
H.~Yu, L.~Zhong, and A.~Sabharwal, ``Beamadapt: Energy efficient beamsteering
  on mobile devices,'' \emph{ACM SIGMOBILE Mob. Comput. Commun. Rev.}, Jul.
  2010.

\bibitem{Ohira2000Electronically}
T.~Ohira and K.~Gyoda, ``Electronically steerable passive array radiator
  antennas for low-cost analog adaptive beamforming,'' in \emph{IEEE
  International Conference on Phased Array Systems and Technology}, 2000.

\bibitem{Sun2014Mimo}
S.~Sun, T.~S. Rappaport, R.~W. Heath, A.~Nix, and S.~Rangan, ``{MIMO} for
  millimeter-wave wireless communications: {B}eamforming, spatial multiplexing,
  or both?'' \emph{IEEE Commun. Mag.}, vol.~52, no.~12, Dec. 2014.

\bibitem{Han2015Large}
S.~Han, C.~L. I, Z.~Xu, and C.~Rowell, ``Large-scale antenna systems with
  hybrid analog and digital beamforming for millimeter wave 5{G},'' \emph{IEEE
  Commun. Mag.}, vol.~53, no.~1, Jan. 2015.

\bibitem{Puglielli2016Design}
A.~Puglielli, A.~Townley, G.~LaCaille, V.~Milovanović, P.~Lu, K.~Trotskovsky,
  A.~Whitcombe, N.~Narevsky, G.~Wright, T.~Courtade, E.~Alon, B.~Nikolić, and
  A.~M. Niknejad, ``Design of energy- and cost-efficient massive {MIMO}
  arrays,'' \emph{Proc. IEEE}, vol. 104, no.~3, Mar. 2016.

\bibitem{Molisch2016Hybrid}
A.~F. Molisch, V.~V. Ratnam, S.~Han, Z.~Li, S.~L.~H. Nguyen, L.~Li, and
  K.~Haneda, ``Hybrid beamforming for massive {MIMO}-a survey,'' \emph{IEEE
  Commun. Mag.}, vol.~55, no.~9, Sep. 2017.

\bibitem{IEEE80211ad}
{IEEE standard}, ``{IEEE} 802.11ad {WLAN} enhancements for very high throughput
  in the 60 {GH}z band,'' Dec. 2012.

\bibitem{IEEE802153c}
------, ``{IEEE} 802.15.3c {WPAN} millimeter-wave-based alternative physical
  layer extension,'' Oct. 2009.

\bibitem{METIS2015}
{METIS Report}, ``Final performance results and consolidated view on the most
  promising multi-node/multi-antenna transmission technologies,'' Feb. 2015.

\bibitem{ITU2015}
{ITU Report}, ``Technical feasibility of {IMT} in bands above 6{GH}z,'' Jul.
  2015.

\bibitem{Keysight2015massive}
{Keysight Technologies}, ``Massive {MIMO} and mm{W}ave technology insight and
  challenges,'' 2015.

\bibitem{Samsung20155g}
{Samsung Electronics}, ``5{G} {V}ision,'' Feb. 2015.

\bibitem{Amitava2016Enabling}
A.~Ghosh, ``Enabling technologies for next generation wireless systems,''
  \emph{Nokia Bell Labs}, Mar. 2016.

\bibitem{Wen2016Bringing}
W.~Tong, ``Bringing 5{G} into reality,'' \emph{Huawei}, Mar. 2016.

\bibitem{Brown2016Promise}
G.~Brown, O.~Koymen, and M.~Branda, ``The promise of 5{G} mm{W}ave - {H}ow do
  we make it mobile?'' \emph{Qualcomm {T}echnologies}, Jun. 2016.

\bibitem{Li2018mobilize}
J.~Li, Y.~Sun, L.~Xiao, S.~Zhou, and A.~Sabharwal, ``How to mobilize mm{W}ave:
  A joint beam and channel tracking approach,'' in \emph{IEEE ICASSP}, Apr.
  2018.

\bibitem{nevel1973stochastic}
M.~B. Nevel'son and R.~Z. Has'minskii, \emph{Stochastic approximation and
  recursive estimation}, 1973.

\bibitem{kushner2003stochastic}
H.~Kushner and G.~G. Yin, \emph{Stochastic approximation and recursive
  algorithms and applications}.\hskip 1em plus 0.5em minus 0.4em\relax
  Springer, 2003.

\bibitem{KaramiLS2007}
E.~Karami, ``Tracking performance of least squares {MIMO} channel estimation
  algorithm,'' \emph{IEEE Trans. Commun.}, vol.~55, no.~11, Nov. 2007.

\bibitem{Gao2015multi}
B.~Gao, Z.~Xiao, L.~Su, Z.~Chen, D.~Jin, and L.~Zeng, ``Multi-device multi-path
  beamforming training for 60-{GH}z millimeter-wave communications,'' in
  \emph{IEEE ICC}, Jun. 2015.

\bibitem{Alkhateeb2015Compressed}
A.~Alkhateeb, G.~Leusz, and R.~W. Heath, ``Compressed sensing based multi-user
  millimeter wave systems: {H}ow many measurements are needed?'' in \emph{IEEE
  ICASSP}, Apr. 2015.

\bibitem{Rial2016Hybrid}
R.~M\'endez-Rial, C.~Rusu, N.~Gonz\'alez-Prelcic, A.~Alkhateeb, and R.~W.
  Heath, ``Hybrid {MIMO} architectures for millimeter wave communications:
  Phase shifters or switches?'' \emph{IEEE Access}, vol.~4, Jan. 2016.

\bibitem{Va2016tracking}
V.~Va, H.~Vikalo, and R.~W. Heath, ``Beam tracking for mobile millimeter wave
  communication systems,'' in \emph{IEEE GlobalSIP}, Dec. 2016.

\bibitem{Lee2014Exploiting}
J.~Lee, G.~T. Gil, and Y.~H. Lee, ``Exploiting spatial sparsity for estimating
  channels of hybrid {MIMO} systems in millimeter wave communications,'' in
  \emph{IEEE GLOBECOM}, Dec. 2014.

\bibitem{Payami2015Effective}
S.~Payami, M.~Shariat, M.~Ghoraishi, and M.~Dianati, ``Effective {RF} codebook
  design and channel estimation for millimeter wave communication systems,'' in
  \emph{IEEE ICC Workshop}, Jun. 2015.

\bibitem{zhu2016auxiliary}
D.~Zhu, J.~Choi, and R.~W. Heath~Jr, ``Auxiliary beam pair enabled {A}o{D} and
  {A}o{A} estimation in closed-loop large-scale millimeter-wave {MIMO}
  system,'' \emph{IEEE Trans. Wireless Commun.}, vol.~16, no.~7, Jul. 2017.

\bibitem{Wang2009Beam}
J.~Wang, Z.~Lan, C.-W. Pyo, T.~Baykas, C.-S. Sum, M.~A. Rahman, J.~Gao,
  R.~Funada, F.~Kojima, H.~Harada, and S.~Kato, ``Beam codebook based
  beamforming protocol for multi-{G}bps millimeter-wave {WPAN} systems,''
  \emph{IEEE J. Sel. Areas Commun.}, vol.~27, no.~8, Oct. 2009.

\bibitem{Hur2013Millimeter}
S.~Hur, T.~Kim, D.~J. Love, J.~V. Krogmeier, T.~A. Thomas, and A.~Ghosh,
  ``Millimeter wave beamforming for wireless backhaul and access in small cell
  networks,'' \emph{IEEE Trans. Commun.}, Oct. 2013.

\bibitem{Alkhateeb2014Channel}
A.~Alkhateeb, O.~E. Ayach, G.~Leus, and R.~W. Heath, ``Channel estimation and
  hybrid precoding for millimeter wave cellular systems,'' \emph{IEEE J. Sel.
  Top. Signal Process.}, vol.~8, no.~5, Oct. 2014.

\bibitem{Alkhateeb2015Limited}
A.~Alkhateeb, G.~Leus, and R.~W. Heath, ``Limited feedback hybrid precoding for
  multi-user millimeter wave systems,'' \emph{IEEE Trans. Wireless Commun.},
  vol.~14, no.~11, Nov. 2015.

\bibitem{zhang2016tracking}
C.~{Zhang}, D.~{Guo}, and P.~{Fan}, ``Tracking angles of departure and arrival
  in a mobile millimeter wave channel,'' in \emph{IEEE ICC}, May 2016.

\bibitem{palacios2016tracking}
J.~Palacios, D.~De~Donno, and J.~Widmer, ``Tracking mm-{W}ave channel dynamics:
  {F}ast beam training strategies under mobility,'' in \emph{IEEE INFOCOM},
  2017.

\bibitem{Gao2016Fast}
X.~Gao, L.~Dai, Y.~Zhang, T.~Xie, X.~Dai, and Z.~Wang, ``Fast channel tracking
  for {T}erahertz beamspace massive {MIMO} systems,'' \emph{IEEE Trans. Veh.
  Technol.}, vol.~66, no.~7, Jul. 2017.

\bibitem{bae2017new}
J.~Bae, S.~H. Lim, J.~H. Yoo, and J.~W. Choi, ``New beam tracking technique for
  millimeter wave-band communications,'' \emph{arXiv preprint
  arXiv:1702.00276}, 2017.

\bibitem{koksal2012robust}
C.~E. Koksal and P.~Schniter, ``Robust rate-adaptive wireless communication
  using {ACK/NAK}-feedback,'' \emph{IEEE Trans. Signal Process.}, vol.~60,
  no.~4, pp. 1752--1765, 2012.

\bibitem{borkar2008stochastic}
V.~S. Borkar, \emph{Stochastic approximation: a dynamical systems viewpoint},
  2008.

\bibitem{Poor1994estimation}
H.~V. Poor, \emph{An introduction to signal detection and estimation}.\hskip
  1em plus 0.5em minus 0.4em\relax New York, NY, USA: Springer-Verlag New York,
  Inc., 1994.

\bibitem{Yin1994rate}
G.~Yin and K.~Yin, ``Asymptotically optimal rate of convergence of smoothed
  stochastic recursive algorithms,'' \emph{Stochastics and Stochastic Reports},
  vol.~47, 1994.

\bibitem{Pedersen20165g}
K.~I. Pedersen, G.~Berardinelli, F.~Frederiksen, P.~Mogensen, and A.~Szufarska,
  ``A flexible 5{G} frame structure design for frequency-division duplex
  cases,'' \emph{IEEE Commun. Mag.}, vol.~54, no.~3, Mar. 2016.

\bibitem{Zong20165g}
P.~Zong, ``5{G} and the path to 5{G},'' \emph{Intel {C}orporation}, Oct. 2016.

\bibitem{Zaidi2016Waveform}
A.~A. Zaidi, R.~Baldemair, H.~Tullberg, H.~Bjorkegren, L.~Sundstrom, J.~Medbo,
  C.~Kilinc, and I.~D. Silva, ``Waveform and numerology to support 5{G}
  services and requirements,'' \emph{IEEE Commun. Mag.}, vol.~54, no.~11, Nov.
  2016.

\bibitem{borkar2002lock}
V.~S. Borkar, ``On the lock-in probability of stochastic approximation,''
  \emph{Combinatorics, Probability \& Computing}, vol.~11, no.~1, 2002.

\bibitem{holte2009discrete}
J.~M. Holte, ``Discrete {G}ronwall lemma and applications,'' in \emph{MAA-NCS
  meeting at the University of North Dakota}, vol.~24, 2009.

\bibitem{Hoeffding1963Probability}
W.~Hoeffding, ``Probability inequalities for sums of bounded random
  variables,'' \emph{J. Am. Stat. Assoc.}, vol.~58, no. 301, 1963.

\end{thebibliography}

\end{document}